\documentclass[12pt]{elsarticle}
\usepackage{eurosym}
\usepackage{mathptm,times,graphicx}

\newcounter{theorem}
\renewcommand\thetheorem{\arabic{section}.\arabic{theorem}}

\newenvironment{theorem}{\par\medskip\noindent\begingroup{\bf Theorem
             \stepcounter{theorem}\thetheorem.}\ \itshape
             \def\@currentlabel{\thetheorem}}{\endgroup\par\medskip}
\newenvironment{proposition}{\par\medskip\noindent\begingroup{\bf Proposition
             \stepcounter{theorem}\thetheorem.}\ \itshape
             \def\@currentlabel{\thetheorem}}{\endgroup\par\medskip}
\newenvironment{remark}{\par\medskip\noindent\begingroup{\bf Remark
             \stepcounter{theorem}\thetheorem.}\
             \def\@currentlabel{\thetheorem}}{\endgroup\par\medskip}

\newenvironment{proof}{\par\noindent{\bf Proof.} }{\proofbox\par\medskip}

\def\proofbox{\hfill{\ensuremath\Box}}

\newcommand{\appeqn}{\setcounter{equation}{0}\renewcommand{\theequation}
        {\mbox{A.\arabic{equation}}}}
\newcommand{\apptitle}
{\renewcommand{\thesection}{Appendix}}
\newcommand{\thank}
{\renewcommand{\thesection}{Acknowledgements}}
\newdimen\LENB \newdimen\LENW \newdimen\THI
\newdimen\LENWH \newdimen\LENTOT \newcount\N
\def\vbrknlnele#1#2#3{
  \LENB=#1pt \LENW=#2pt \THI=#3pt
  \LENWH=\LENW \divide\LENWH by 2
  \LENTOT=\LENB \advance\LENTOT by \LENW
  \vbox to \LENTOT{
    \vbox to \LENWH{}
    \nointerlineskip
    \vbox to \LENB{\hbox to \THI{\vrule width \THI height \LENB}}
    \nointerlineskip
    \vbox to \LENWH{}
  }}
\def\vbrknln#1{
  \N=#1
  \vcenter{
    \vbox{
      \loop\ifnum\N>0
        \vbox to 4pt{\vbrknlnele{2}{2}{0.1}}
        \nointerlineskip
        \advance\N by -1
      \repeat
  }}}

\def\hbrknlnele#1#2#3{
  \LENB=#1pt \LENW=#2pt \THI=#3pt
  \LENTOT=\LENB \advance\LENTOT by \LENW
  \vcenter{
    \vbox to \THI{
      \hbox to \LENTOT{
        \hfil
        \vrule width \LENB height \THI
        \hfil}
  }}}

\makeatletter

\begin{document}

\begin{frontmatter}
\title{Complex short pulse and coupled complex short pulse equations}
\author[BF]{Bao-Feng Feng}
\address[BF]{
Department of Mathematics \\ The University of Texas-Pan
American\\  Edinburg, TX, 78541-2999, USA}


\begin{abstract}
In the present paper, we propose a complex short pulse equation and a coupled complex short
equation to describe ultra-short pulse propagation in optical fibers.
They are integrable due to the existence of Lax pairs and infinite number of conservation laws.
Furthermore, we find their multi-soliton solutions in terms of pfaffians by virtue of Hirota's bilinear
method. One- and two-soliton solutions are investigated in details, showing favorable properties
in modeling ulta-short pulses with a few optical cycles. Especially, same as the coupled nonlinear Schr\"odinger equation, there is an interesting phenomenon of energy redistribution in soliton interactions. It is expected that, for the ultra-short pulses, the complex and coupled complex short pulses equation will play the same roles as the nonlinear Schr\"odinger equation and coupled nonlinear Schr\"odinger equation.
\vspace{1cm}

\end{abstract}
\end{frontmatter}



\section{Introduction}
The nonlinear Schr\"odinger (NLS) equation, as one of the universal
equations that describe the evolution of slowly varying packets of
quasi-monochromatic waves in weakly nonlinear dispersive media, has been
very successful in many applications such as nonlinear optics and water
waves \cite{Kodamabook, Agrawalbook,Boydbook,Yarivbook}. The NLS equation is
integrable, which can be solved by the inverse scattering transform \cite%
{Zakharov}.

However, in the regime of ultra-short pulses where the width of optical
pulse is in the order of femtosecond ($10^{-15}$ s), the NLS equation
becomes less accurate \cite{Rothenberg}. Description of ultra-short
processes requires a modification of standard slow varying envelope models
based on the NLS equation. There are usually two approaches to meet this
requirement in the literature. The first one is to add several higher-order
dispersive terms to get higher-order NLS equation \cite{Agrawalbook}. The
second one is to construct a suitable fit to the frequency-dependent
dielectric constant $\epsilon(\omega)$ in the desired spectral range.
Several models have been proposed by this approach including the short-pulse (SP) equation \cite{SPE_Org,KimPRL,KimPRA,Bandelow}.

Recently, Sch\"{a}fer and Wayne derived a so-called short pulse (SP)
equation \cite{SPE_Org}
\begin{equation}
u_{xt}=u+ \frac 16 \left(u^{3}\right)_{xx}\,,  \label{SPE}
\end{equation}%
to describe the propagation of ultra-short optical pulses in nonlinear
media. Here, $u=u(x,t)$ is a real-valued function, representing the
magnitude of the electric field, the subscripts $t$ and $x$ denote partial
differentiation. Apart from the context of nonlinear optics, the SP equation
has also been derived as an integrable differential equation associated with
pseudospherical surfaces \cite{Robelo}. The SP equation has been shown to be
completely integrable \cite{Robelo,Beals,Sakovich,Brunelli1,Brunelli2}. The
periodic and soliton solutions of the SP equation were found in \cite%
{Sakovich2,Kuetche,Parkes}. The connection between the SP equation and the
sine-Gordon equation through the hodograph transformation was clarified, and
then the $N$-soliton solutions including multi-loop and multi-breather ones
were given in \cite{Matsuno_SPE,Matsuno_SPEreview} by using the Hirota
bilinear method \cite{Hirota}. The integrable discretization of the SP
equation was studied in \cite{SPE_discrete1}, the geometric interpretation
of the SP equation, as well as its integrable discretization, was given in
\cite{SPE_discrete2}. The higher-order corrections to the SP equation was
studied in \cite{Schafer13} most recently.

Similar to the case of the NLS equation \cite{Manakov1974}, it is necessary
to consider its two-component or multi-component generalizations of the SP
equation for describing the effects of polarization or anisotropy. As a
matter of fact, several integrable coupled short pulse have been proposed in
the literature \cite%
{PKB_CSPE,Sakovich3,Hoissen_CSPE,Matsuno_CSPE,Feng_CSPE,ZengYao_CSPE}. Most
recently, the bi-Hamiltonian structures for the above two-component SP
equations were obtained by Brunelli \cite{Brunelli_CSPE}.


In the present paper, we propose and study a complex short pulse (CSP)
equation
\begin{equation}
q_{xt}+q+ \frac{1}{2} \left(|q|^{2}q_x \right)_{x}=0\,,  \label{CSP}
\end{equation}%
and its two-component generalization
\begin{eqnarray}
&& q_{1,xt}+q_1+\frac{1}{2} \left((|q_1|^2+|q_2|^2)q_{1,x}\right)_{x}=0\,,
\label{CCSP1} \\
&& q_{2,xt}+q_2+\frac{1}{2} \left((|q_1|^2+|q_2|^2)q_{2,x}\right)_{x}=0\,.
\label{CCSP2}
\end{eqnarray}
As will be revealed in the present paper, both the CSP equation and its
two-component generalization are integrable guaranteed by the existence of
Lax pairs and infinite number of conservation laws. They have $N$-soliton
solutions which can be constructed via Hirota's bilinear method.


The outline of the present paper is organized as follows. In section 2, we
derive the CSP equation and coupled complex short pulse (CCSP) equation from
the physical context. In section 3, by providing the Lax pairs, the
integrability of the proposed two equations are confirmed, and further, the
conservation laws, both local and nonlocal ones, are investigated. Then $N$%
-soliton solutions to both the CSP and CCSP equations are constructed in
terms of pfaffians by Hirota's bilinear method in section 4. In section 5,
soliton-interaction for coupled complex short pulse equation is investigated
in details, which shows rich phenomena similar to the 
coupled nonlinear Schr\"odinger equatoin. In particular, they may undergo
either elastic or inelastic collision depending on the initial conditions.
For inelastic collisions, there is an energy exchange between solitons,
which can allow the generation or vanishing of soliton. The dynamics is more
richer in compared with the single component case. The paper is concluded by
comments and remarks in section 6.

\section{The derivation of the complex short pulse and coupled complex short
pulse equations}

In this section, following the procedure in \cite{Agrawalbook,SPE_Org}, we
derive the complex short pulse equation (\ref{CSP}) and its two-component
generalization that governs the propagation of ultra short pulse packet
along optical fibers.

\subsection{The complex short pulse equation}

We start with a wave equation for electric field
\begin{equation}  \label{E-wave-equation}
\nabla^2 \mathbf{E} -\frac{1}{c^2} \mathbf{E}_{tt} = \mu_0 \mathbf{P}_{tt}\,,
\end{equation}
originated from the Maxwell equation. Here $\mathbf{E} ( \mathbf{r},t)$ and $%
\mathbf{P} ( \mathbf{r},t)$ represent the electric field and the induced
polarization, respectively, $\mu_0$ is the vacuum permeability, $c$ is the
speed of light in vacuum. If we assume the local medium response and only
the third-order nonlinear effects governed by $\chi^{(3)}$, the induced
polarization consists of two parts, $\mathbf{P} ( \mathbf{r},t)=\mathbf{P}%
_{L} ( \mathbf{r},t)+\mathbf{P}_{NL} ( \mathbf{r},t)$, where the linear part
\begin{equation}  \label{P_L}
\mathbf{P}_{L} ( \mathbf{r},t)=\epsilon_0 \int_{-\infty}^{\infty} \chi^{(1)}
(t-t^{\prime })\cdot \mathbf{E} ( \mathbf{r},t^{\prime })\,dt^{\prime }\,,
\end{equation}
and the nonlinear part
\begin{equation}  \label{P_NL}
\mathbf{P}_{NL}( \mathbf{r},t)=\epsilon_0 \int_{-\infty}^{\infty} \chi^{(3)}
(t-t_1,t-t_2,t-t_3)\times \mathbf{E} ( \mathbf{r},t_1) \mathbf{E} ( \mathbf{r%
},t_2) \mathbf{E} ( \mathbf{r},t_3)\,dt_1dt_2dt_3\,.
\end{equation}
Here $\epsilon_0$ is the vacuum permittivity and $\chi^{(j)}$ is the $j$%
th-order susceptibility. Since the nonlinear effects are relatively small in
silica fibers, $\mathbf{P}_{NL}$ can be treated as a small perturbation.
Therefore, we first consider (\ref{E-wave-equation}) with $\mathbf{P}_{NL}=0$%
. Furthermore, we restrict ourselves to the case that the optical pulse
maintains its polarization along the optical fiber, and the transverse
diffraction term $\Delta_{\perp} \mathbf{E}$ can be neglected. In this case,
the electric field can be considered to be one-dimensional and expressed as
\begin{equation}  \label{E-field}
\mathbf{E} = \frac 12 \mathbf{e_1} \left(E(z,t)+c.c. \right)\,,
\end{equation}
where $\mathbf{e_1}$ is a unit vector in the direction of the polarization,
$E(z,t)$ is the complex-valued function, and $c.c.$ stands for the complex
conjugate. Conducting a Fourier transform on (\ref{E-wave-equation}) leads
to the Helmholtz equation
\begin{equation}  \label{Waveequation-Fourier}
\tilde {E}_{zz} (z,\omega) + \epsilon(\omega) \frac{\omega^2}{c^2} \tilde {E}
(z,\omega)=0\,,
\end{equation}
where $\tilde {E} (z,\omega)$ is the Fourier transform of $E(z,t)$ defined
as
\begin{equation}  \label{E_Fourier}
\tilde E (z,\omega)=\int_{-\infty}^{\infty} {\ E} (z,t) e^{\mathrm{i} \omega
t}\,dt\,,
\end{equation}
$\epsilon(\omega)$ is called the frequency-dependent dielectric constant
defined as
\begin{equation}  \label{Dielectric}
\epsilon(\omega)=1+ \tilde \chi^{(1)} (\omega)\,,
\end{equation}
where $\tilde \chi^{(1)} (\omega)$ is the Fourier transform of $\chi^{(1)}
(t)$
\begin{equation}  \label{Chi_Fourier}
\tilde \chi^{(1)} (\omega)=\int_{-\infty}^{\infty} \chi^{(1)} (t) e^{\mathrm{%
i} \omega t}\,dt\,.
\end{equation}
Now we proceed to the consideration of the nonlinear effect. Assuming the
nonlinear response is instantaneous so that $P_{NL}$ is given by $%
P_{NL}(z,t)= \epsilon_0 \epsilon_{NL} E(z,t)$ \cite{Agrawalbook} where the
nonlinear contribution to the dielectric constant is defined as
\begin{equation}  \label{Epsnl}
\epsilon_{NL}= \frac 34 \chi^{(3)}_{xxxx} |E(z,t)|^2\,.
\end{equation}
In this case, the Helmholtz equation (\ref{Waveequation-Fourier}) can be
modified as
\begin{equation}  \label{Helmholtz}
\tilde {E}_{zz} (z,\omega) + \tilde \epsilon(\omega) \frac{\omega^2}{c^2}
\tilde {E} (z,\omega)=0\,,
\end{equation}
where
\begin{equation}  \label{Modified_Dielectric}
\tilde \epsilon(\omega)=1+ \tilde \chi^{(1)} (\omega)+ \epsilon_{NL}\,.
\end{equation}
As pointed out in \cite{SPE_Org,Boydbook,KimPRL}, the Fourier transform $%
\tilde \chi^{(1)}$ can be well approximated by the relation $\tilde
\chi^{(1)}=\tilde \chi_0^{(1)} - \tilde \chi_2^{(1)} \lambda^2$ if we
consider the propagation of optical pulse with the wavelength between 1600
nm and 3000 nm. It then follows that the linear equation (\ref%
{Waveequation-Fourier}) written in Fourier transformed form becomes
\begin{equation}  \label{Wave_Fourier2}
\tilde {E}_{zz} + \frac{1+\tilde \chi_0^{(1)}}{c^2} \omega^2 \tilde {E} -
(2\pi)^2 \tilde \chi_2^{(1)} \tilde {E} + \epsilon_{NL} \frac{\omega^2}{c^2}
\tilde {E} =0\,.
\end{equation}

Applying the inverse Fourier transform to (\ref{Wave_Fourier2}) yields a
single nonlinear wave equation
\begin{equation}  \label{Wave_nonlinear1}
E_{zz} - \frac{1}{c_1^2} E_{tt} = \frac{1}{c_2^2} E +\frac 34
\chi^{(3)}_{xxxx} \left(|E|^2 E\right)_{tt} =0\,.
\end{equation}

Similar to \cite{SPE_Org}, we focus on only a right-moving wave packet and
make a multiple scales \textrm{ansatz}
\begin{equation}  \label{E_ansatz}
E(z,t)=\epsilon E_0(\phi, z_1, z_2, \cdots)+ \epsilon^2 E_1(\phi, z_1, z_2,
\cdots) + \cdots\,,
\end{equation}
where $\epsilon$ is a small parameter, $\phi$ and $z_n$ are the scaled
variables defined by
\begin{equation}  \label{Multiple_ansatz}
\phi= \frac{t-\frac{x}{c_1}}{\epsilon}, \quad z_n=\epsilon^n z\,.
\end{equation}
Substituting (\ref{E_ansatz}) with (\ref{Multiple_ansatz}) into (\ref%
{Wave_nonlinear1}), we obtain the following partial differential equation
for $E_0$ at the order $O(\epsilon)$:

\begin{equation}  \label{Wave_nonlinear2}
- \frac{2}{c_1} \frac{\partial^2 E_0}{\partial{\phi}\partial{z_1}} = \frac{1%
}{c_2^2} E_0 +\frac 34 \chi^{(3)}_{xxxx} \frac{\partial}{\partial{\phi}}
\left(|E_0|^2 \frac{\partial E_0}{\partial {\phi}} \right)\,.
\end{equation}


Finally, by a scale transformation
\begin{equation}  \label{Scaling}
x= \frac{c_1}{2} \phi, \quad t={c_2} z_1, \quad q = \frac{c_1\sqrt{%
6c_2\chi^{(3)}_{xxxx}}}{4} E_0\,,
\end{equation}
we arrive at the normalized form of the complex short pulse equation (\ref%
{CSP}).

\subsection{Coupled complex short pulse equation}

In the previous subsection, a major simplification made in the derivation of
the complex short pulse equation is to assume that the polarization is
preserved during its propagating inside an optical fiber. However, this is
not really the case in practice. For birefringent fibers, two orthogonally
polarized modes have to be considered. Therefore, similar to the extension
of coupled nonlinear Schr{\"o}dinger equation from the NLS equation, an
extension to a two-component version of the complex short pulse equation (%
\ref{CSP}) is needed to describe the propagation of ultra-short pulse in
birefringent fibers. In fact, several generalizations have been proposed for
the short pulse equation \cite%
{PKB_CSPE,Sakovich3,Hoissen_CSPE,Matsuno_CSPE,Feng_CSPE,ZengYao_CSPE}.
Particularly, by taking into account the effects of anisotropy and
polarization, Pietrzyk \textit{et. al.} have derived a general two-component
short-pulse equation from the physical context \cite{PKB_CSPE}. We follow
the approach by Pietrzyk \textit{et. al.} to derive a two-component complex
short pulse equation. However, as shown in subsequent section, the
two-component complex short pulse equation admits multi-soliton solutions
which reveals richer dynamics in soliton interactions in compared with the
real SP equation.

We first consider the linear birefringent fiber such that the electric field
with an arbitrarily polarized optical fiber can be expressed as
\begin{equation}  \label{E-field2}
\mathbf{E} = \frac 12 \left( \mathbf{e_1} E_1(z,t)+ \mathbf{e_2} E_2(z,t)
\right)+ c.c.\,,
\end{equation}
where $\mathbf{e_1}$, $\mathbf{e_2}$ are two unit vectors along positive $x$%
- and $y$-direction in the transverse plane perpendicular to the optical
fiber, respectively, $E_1$ and $E_2$ are the complex amplitudes of the
polarization components correspondingly. Without the presence of nonlinear
polarization ($P_{NL}=0$) and the transverse diffraction, the Fourier
transform converts (\ref{E-wave-equation}) into a pair of Helmholtz
equations
\begin{equation}  \label{CCSPE_Helmholtz1}
\tilde {E}_{1,zz} (z,\omega) + \epsilon(\omega) \frac{\omega^2}{c^2} \tilde {%
E_1} (z,\omega)=0\,,
\end{equation}

\begin{equation}  \label{CCSPE_Helmholtz2}
\tilde {E}_{2,zz} (z,\omega) + \epsilon(\omega) \frac{\omega^2}{c^2} \tilde {%
E_2} (z,\omega)=0\,.
\end{equation}
Same as the scalar case, the frequency-dependent dielectric constant $%
\epsilon(\omega)=1+ \tilde \chi^{(1)} (\omega)$, where $\tilde \chi^{(1)}$
can be well approximated by the relation $\tilde \chi^{(1)}=\tilde
\chi_0^{(1)} - \tilde \chi_2^{(1)} \lambda^2$ for the propagation of optical
pulse with the wavelength between 1600 nm and 3000 nm.

As indicated in \cite{Agrawalbook}, the nonlinear part of the induced
polarization $\mathbf{P}_{NL}$ can be written as
\begin{equation}
\mathbf{P}_{NL}=\frac{1}{2}\left( \mathbf{e_{1}}P_{1}(z,t)+\mathbf{e_{2}}%
P_{2}(z,t)\right) +c.c.\,,  \label{E-field22}
\end{equation}%
where
\begin{equation}
P_{1}=\frac{3\epsilon _{0}}{4}\chi _{xxxx}^{(3)}\left[ \left( |E_{1}|^{2}+%
\frac{2}{3}|E_{2}|^{2}\right) E_{1}+\frac{1}{3}(E_{1}^{\ast }E_{2})E_{2}%
\right] \,,  \label{Nonlinear_Polarization1}
\end{equation}%
\begin{equation}
P_{2}=\frac{3\epsilon _{0}}{4}\chi _{xxxx}^{(3)}\left[ \left( |E_{2}|^{2}+%
\frac{2}{3}|E_{1}|^{2}\right) E_{2}+\frac{1}{3}(E_{2}^{\ast }E_{1})E_{1}%
\right] \,.  \label{Nonlinear_Polarization2}
\end{equation}%
The last term in Eqs. (\ref{Nonlinear_Polarization1}) and (\ref%
{Nonlinear_Polarization2}) leads to the degenerate four-wave mixing. In
highly birefringent fibers, the four-wave-mixing term can often be
neglected. In this case, we arrive at a coupled nonlinear wave equation
\begin{equation}  \label{Coupled_equations1}
E_{1,zz}-\frac{1}{c_{1}^{2}}E_{1,tt}=\frac{1}{c_{2}^{2}}E_{1}+\frac{3}{4}%
\chi _{xxxx}^{(3)}\left[ \left( |E_{1}|^{2}+\frac{2}{3}|E_{2}|^{2}\right)
E_{1}\right] _{tt}\,,
\end{equation}

\begin{equation}  \label{Coupled_equations2}
E_{2,zz}-\frac{1}{c_{1}^{2}}E_{2,tt}=\frac{1}{c_{2}^{2}}E_{2}+\frac{3}{4}%
\chi _{xxxx}^{(3)}\left[ \left( |E_{2}|^{2}+\frac{2}{3}|E_{1}|^{2}\right)
E_{2}\right] _{tt}\,.
\end{equation}

Similar to the scalar case, by a multiple scales expansion and an
appropriate scaling transformation, a couple complex short pulse equation
can be obtained from (\ref{Coupled_equations1})--(\ref{Coupled_equations2})
\begin{equation}  \label{CCSP_linear1}
q_{1,xt}+q_{1}+\frac{1}{2} \left((|q_{1}|^{2}+\frac{2}{3}|q_{2}|^{2})q_{1,x}%
\right) _{x}=0\,,
\end{equation}

\begin{equation}  \label{CCSP_linear2}
q_{2,xt}+q_{2}+\frac{1}{2} \left((|q_{2}|^{2}+\frac{2}{3}|q_{1}|^{2})q_{2,x}%
\right) _{x}=0\,.
\end{equation}

More generally, we can consider the coupled short pulse equation for
elliptically birefringent fibers. In this case, the electric field can be
written as
\begin{equation}
\mathbf{E}=\frac{1}{2}\left( \mathbf{e_{x}}E_{x}(z,t)+\mathbf{e_{y}}%
E_{y}(z,t)\right) +c.c.\,,  \label{E-field3}
\end{equation}%
where $\mathbf{e_{x}}$ and $\mathbf{e_{y}}$ are orthonormal polarization
eigenvectors
\begin{equation}
\mathbf{e_{x}}=\frac{\mathbf{e_{1}}+ir\mathbf{e_{2}}}{\sqrt{1+r^{2}}},\quad
\mathbf{e_{y}}=\frac{r\mathbf{e_{1}}-i\mathbf{e_{2}}}{\sqrt{1+r^{2}}}\,.
\label{Elliptical-unit}
\end{equation}%
The parameter $r$ represents the ellipticity. It is common to introduce the
ellipticity angle $\theta $ as $r=\tan (\theta /2)$. The case $\theta =0$
and $\pi /2$ correspond to linearly and circularly birefringent fibers,
respectively.

Following a procedure similar to the case of linearly birefringent fibers,
one can drive the normalized form for the coupled complex short pulse
equation
\begin{equation}  \label{CCSP_general1}
q_{1,xt}+ q_1+ \frac 12 \left((|q_1|^2+ B|q_2|^2) q_{1,x} \right)_x = 0\,,
\end{equation}
\begin{equation}  \label{CCSP_general2}
q_{2,xt}+ q_2+ \frac 12  \left((|q_2|^2+ B|q_1|^2) q_{2,x} \right)_x = 0\,.
\end{equation}
where the parameter $B$ is related to the ellipticity angle $\theta$ as
\begin{equation}  \label{B_theta}
B=\frac{2+2\sin^2 \theta}{2+\cos^2 \theta}\,.
\end{equation}
For a linearly birefringent fiber ($\theta=0$), $B=\frac 23$, and Eqs. (\ref%
{CCSP_general1})-- (\ref{CCSP_general2}) reduces to Eqs. (\ref{CCSP_linear1}%
)-- (\ref{CCSP_linear2}). For a circularly birefringent fiber ($\theta=\pi/2$%
), $B=2$. In general, the coupling parameter $B$ depends on the ellipticity
angle $\theta$ and can vary from $\frac 23$ to $2$ for values of $\theta$ in
the range from $0$ to $\pi/2$. Note that $B=1$ when $\theta \approx 35^\circ$%
. As discussed in the subsequent section, this case is of particular
interest because the coupled system is integrable and admits $N$-soliton
solution. 
\section{Lax pairs and conservation laws for the complex and
coupled complex short pulse equations}

\subsection{Lax pairs and integrability}
In \cite{PKB_CSPE}, a matrix generalization for the SP equation is given
based on zero-curvature representation, from which the Lax pairs for several
integrable two-component SP equations are explicitly provided. In this
subsection, we will show the integrability of the complex short pulse and
coupled complex short pulse equations by finding their Lax pairs constructed
from another matrix generalization of the SP equation.

The Lax pair for the complex short pulse equation (\ref{CSP}) can be
expressed as
\begin{equation}  \label{comSPE_Laxa}
\Psi_x=U \Psi, \quad \Psi_t=V\Psi\,,
\end{equation}
with
\begin{equation}  \label{comSPE_Laxb}
\displaystyle U= \lambda \left(%
\begin{array}{cc}
1 & q_x \\
q^{*}_x & -1%
\end{array}%
\right), \quad V= \left(%
\begin{array}{cc}
-\frac {\lambda}{2} |q|^2-\frac{1}{4\lambda} & -\frac{\lambda}{2}%
|q|^2q_x+\frac q2 \\
-\frac{\lambda}{2}|q|^2q^{*}_x-\frac {q^{*}}{2} & \frac {\lambda}{2}|q|^2+\frac{1}{4\lambda}%
\end{array}%
\right)\,.
\end{equation}
It can be easily shown that the compatibility condition $U_t-V_x+[U,\,V]=0$
gives the complex short pulse equation (\ref{CSP}).

The Lax pair for the coupled complex short pulse equation (\ref{CCSP1})--(%
\ref{CCSP2}) is found to be of the form:
\begin{equation}
\Psi _{x}=U\Psi ,\quad \Psi _{t}=V\Psi \,,  \label{ccomSPE_Laxa}
\end{equation}%
with
\begin{equation}
U=\lambda \left(
\begin{array}{cc}
I_{2} & Q_{x} \\
R_{x} & -I_{2}%
\end{array}%
\right) ,\quad   V=\left(
\begin{array}{cc}
-\frac{\lambda }{2}QR-\frac{1}{4\lambda }I_{2} & -\frac{\lambda }{2}QRQ_{x}+%
\frac{1}{2}Q \\
-\frac{\lambda }{2}RQR_{x}-\frac{1}{2}R & \frac{\lambda }{2}QR+\frac{1}{%
4\lambda }I_{2}%
\end{array}%
\right) \,,  \label{ccomSPE_Laxb}
\end{equation}%
where $I_{2}$ is a $2\times 2$ identity matrix, $Q$, $R$ are $2\times 2$
matrices defined as
\begin{equation}
Q=\left(
\begin{array}{cc}
q_{1} & q_{2} \\
-q_{2}^{\ast } & q_{1}^{\ast }%
\end{array}%
\right) ,\quad R=\left(
\begin{array}{cc}
q_{1}^{\ast } & -q_{2} \\
q_{2}^{\ast } & q_{1}%
\end{array}%
\right) \,.  \label{ccomSPE_Laxc}
\end{equation}%
Note that $R=Q^{\dag }$, thus,
\begin{equation}
QR=RQ=(|q_{1}|^{2}+|q_{2}|^{2})I_{2}\,,
\end{equation}%
the compatibility condition $U_{t}-V_{x}+[U,\,V]=0$ for (\ref{ccomSPE_Laxa})
gives the coupled complex short pulse equation (\ref{CCSP1})--(\ref{CCSP2}).

As a matter of fact, the coupled complex short pulse equation can be
generalized into a multi-component, or a vector complex short pulse equation
\begin{equation}
q_{i,xt}+q_{i}+\frac{1}{2}\left( |\mathbf{q}|^{2}q_{i,x}\right)
_{x}=0\,,\quad i=1,\cdots ,n,  \label{NCSPE}
\end{equation}%
where $\mathbf{q}=(q_1,q_2, \cdots, q_n)$. The integrability of Eq. (\ref%
{NCSPE}) 
can be guaranteed by the Lax pair constructed in a similar way as in \cite%
{TsuchidaJPSJ}.
\begin{equation}
\Psi _{x}=U\Psi ,\quad \Psi _{t}=V\Psi \,,  \label{vcomSPE_Laxa}
\end{equation}%
with
\[
U=\lambda \left(
\begin{array}{cc}
I_{2^{n-1}} & Q_{x}^{(n)} \\
R_{x}^{(n)} & -I_{2^{n-1}}%
\end{array}%
\right) ,
\]%
\[
V=\left(
\begin{array}{cc}
-\frac{1}{2}Q^{(n)}R^{(n)}-\frac{1}{4\lambda }I_{2^{n-1}} & -\frac{\lambda }{2}%
Q^{(n)}R^{(n)}Q_{x}^{(n)}+\frac{1}{2}Q^{(n)} \\
-\frac{\lambda }{2}R^{(n)}Q^{(n)}R_{x}^{(n)}-\frac{1}{2}R^{(n)} & \frac{1}{2}%
Q^{(n)}R^{(n)}+\frac{1}{4\lambda }I_{2^{n-1}}%
\end{array}%
\right) \,,
\]%
where $I_{2^{n-1}}$ is a $2^{n-1}\times 2^{n-1}$ identity matrix, $Q^{(n)}$
and $R^{(n)}$ are $2^{n-1}\times 2^{n-1}$ matrices can be constructed
recursively as follows
\begin{equation}
Q^{(1)}=q_{1},\quad R^{(1)}=q_{1}^{\ast }\,,  \label{NCSPE_Laxb1}
\end{equation}%
\begin{equation}
Q^{(n+1)}=\left(
\begin{array}{cc}
Q^{(n)} & q_{n+1}I_{2^{n-1}} \\
-q_{n+1}^{\ast }I_{2^{n-1}} & R^{(n)}%
\end{array}%
\right) \,,  \label{NCSPE_Laxb2}
\end{equation}

\begin{equation}  \label{NCSPE_Laxb3}
R^{(n+1)}= \left(%
\begin{array}{cc}
R^{(n)} & -q_{n+1} I_{2^{n-1}} \\
q^*_{n+1}I_{2^{n-1}} & Q^{(n)}%
\end{array}%
\right) \,.
\end{equation}
By the above construction, we have $R^{(n+1)}=(Q^{(n+1)})^\dag$, and further
\begin{equation}
Q^{(n)} R^{(n)}=R^{(n)}Q^{(n)}=\sum_{i=1}^n|q_i|^2I_{2^{n-1}}\,.
\end{equation}
Therefore, the zero curvature condition $U_t-V_x+[U,\,V]=0$ gives the vector
complex coupled short pulse equation (\ref{NCSPE}). 

\subsection{Local and nonlocal conservation laws}

Following a systematic method developed by in \cite%
{TsuchidaJPSJ,WadatiPTP75,WadatiJPSJ79,Zimerman}, we construct conservation
laws for the vector complex short pulse equation, the conservation laws for
the complex and coupled short pulse equations can be treated as special
cases for $n=1,2$, respectively. To this end, let us rewrite the Lax pair
for the vector complex short pulse equation as follows:
\begin{equation}  \label{Laxpair_vcspe1}
\left(%
\begin{array}{c}
\Psi_1 \\
\Psi_2%
\end{array}
\right)_x = \left(%
\begin{array}{cc}
\lambda I & \lambda Q_x \\
\lambda R_x & - \lambda I%
\end{array}%
\right) \left(%
\begin{array}{c}
\Psi_1 \\
\Psi_2%
\end{array}
\right)\,,
\end{equation}

\begin{equation}  \label{Laxpair_vcspe2}
\left(%
\begin{array}{c}
\Psi_1 \\
\Psi_2%
\end{array}
\right)_t = \left(\begin{array}{cc}
-\frac{\lambda }{2}QR-\frac{1}{4\lambda }I_{2} & -\frac{\lambda }{2}QRQ_{x}+%
\frac{1}{2}Q \\
-\frac{\lambda }{2}RQR_{x}-\frac{1}{2}R & \frac{\lambda }{2}QR+\frac{1}{%
4\lambda }I_{2}%
\end{array}%
\right)  \left(%
\begin{array}{c}
\Psi_1 \\
\Psi_2%
\end{array}
\right)\,.
\end{equation}
Here the size of matrices in the entries of Eqs. (\ref{Laxpair_vcspe1})--(%
\ref{Laxpair_vcspe2}) is of $2^{n-1} \times 2^{n-1}$ and is omitted for
brevity. If we define
\begin{equation}  \label{Cons_law1}
\Gamma \equiv \Psi_2 \Psi_1^{-1}\,
\end{equation}
then we have
\begin{equation}  \label{Cons_law2}
2 \lambda Q_x \Gamma = \lambda Q_x R_x - Q_x ((Q_x)^{-1} \cdot Q_x \Gamma)_x
- \lambda (Q_x \Gamma)^2\,
\end{equation}


Expanding $Q_x \Gamma$ in terms of the spectral parameter $\lambda$ as
follows
\begin{equation}  \label{Cons_law4}
Q_x \Gamma = \sum_{n=0}^\infty F_n \lambda^{-n}\,,
\end{equation}
and substituting into Eq. (\ref{Cons_law2}), we obtain the following
relation
\begin{equation}  \label{Cons_law5}
2 \lambda F_n = Q_x R_x \delta_{n,0} - Q_x((Q_x)^{-1} F_{n-1})_x -
\sum_{l=0}^n F_l F_{n-l}.
\end{equation}
The first local conserved density turns out to be
\begin{equation}  \label{Cons_law6}
F_0=\left(-1+\sqrt{1+\sum|q_{i,x}|^2}\right)I\,,
\end{equation}

which is associated with a Hamiltonian of
\begin{equation}  \label{Cons_law7a}
H_0=\int \sqrt{1+|q_x|^2}\, dx \,,
\end{equation}
for the complex short pulse equation (\ref{CSP}) and
\begin{equation}  \label{Cons_law8a}
H_0=\int \sqrt{1+|q_{1,x}|^2+|q_{2,x}|^2}\, dx \,,
\end{equation}
for the coupled complex short pulse equation (\ref{CCSP1})--(\ref{CCSP2}).

Following the procedure in \cite{Zimerman}, we can find the nonlocal
conservation laws for vector complex short pulse equation. To this end, we
expand $Q_x \Gamma$ as follows
\begin{equation}  \label{Cons_law9}
Q_x \Gamma = \sum_{n=1}^\infty F_{-n} (2\lambda)^{n}\,.
\end{equation}
The first two orders in $\lambda$ yield the following equations
\begin{equation}  \label{Cons_law10}
0 = Q_x R_x -Q_x \left((Q_x)^{-1}F_{-1}\right)_x\,,
\end{equation}
\begin{equation}  \label{Cons_law11}
2F_{-1} = -Q_x \left((Q_x)^{-1}F_{-2}\right)_x\,,
\end{equation}
from which, the first two nonlocal conserved densities can be calculated as
\begin{equation}  \label{Cons_law12}
F_{-1}= \frac 12 Q_x R\,,
\end{equation}
\begin{equation}  \label{Cons_law13}
F_{-2}= \frac 12 Q R\ -\frac 12 \partial_x\left(Q\partial_x R \right).
\end{equation}
The first one turns out to be a trivial one, the second one accounts for a
Hamiltonian
\begin{equation}  \label{Cons_law7c}
H_{-1}= \frac 12 \int |q|^2\, dx \,,
\end{equation}
for the complex short pulse equation (\ref{CSP}) and
\begin{equation}  \label{Cons_law8c}
H_{-1}= \frac 12\int (|q_{1}|^2+|q_{2}|^2)\, dx \,,
\end{equation}
for the coupled complex short pulse equation (\ref{CCSP1})--(\ref{CCSP2}).
\section{Multi-soliton solutions by Hirota's bilinear method}


\subsection{Bilinear equations and $N$-soliton solution to the complex short
pulse equation}

\begin{proposition}
The complex short pulse equation is derived from the following bilinear
equations.
\begin{equation}  \label{CSPE_bilinear1}
D_sD_y f \cdot g =fg\,,
\end{equation}
\begin{equation}  \label{CSPE_bilinear2}
D^2_s f \cdot f =\frac{1}{2} |g|^2\,,
\end{equation}
by dependent variable transformation 
\begin{equation}  \label{CSP_dtrf}
q=\frac{g}{f}\,,
\end{equation}
and hodograph transformation
\begin{equation}
x = y -2(\ln f)_s\,, \quad t=-s \,,  \label{CSP_hodograph}
\end{equation}
where $D$ is called Hirota $D$-operator defined by
\[
D_s^n D_y^m f\cdot g=\left(\frac{\partial}{\partial s} -\frac{\partial}{%
\partial s^{\prime }}\right)^n \left(\frac{\partial}{\partial y} -\frac{%
\partial}{\partial y^{\prime }}\right)^m f(y,s)g(y^{\prime },s^{\prime
})|_{y=y^{\prime }, s=s^{\prime }}\,.
\]
\end{proposition}

.

\begin{proof}
Dividing both sides by $f^2$, the bilinear equations (\ref{CSPE_bilinear1}%
)-- (\ref{CSPE_bilinear2}) can be cast into
\begin{equation}
\left\{%
\begin{array}{l}
\displaystyle \left(\frac{g}{f} \right)_{sy} + 2\frac{g}{f} \left( \ln
f\right)_{sy} = \frac{g}{f}\,, \\[5pt]
\displaystyle \left( \ln f\right)_{ss} =\frac{1}{4} \frac{|g|^2}{f^2}\,.%
\end{array}%
\right.  \label{CSP_BL2}
\end{equation}
From the hodograph transformation and dependent variable transformation, we
then have
\[
\frac{\partial x}{\partial s} = -2(\ln f)_{ss} = -\frac 12 |q|^2\,, \qquad
\frac{\partial x}{\partial y} = 1-2(\ln f)_{sy}\,,
\]
which implies
\begin{equation}  \label{CSP_BL3}
{\partial_y} = \rho^{-1} {\partial_x}\,, \qquad {\partial_s} = -{\partial_t}
- \frac 12 |q|^2 {\partial_x}\,
\end{equation}
by letting $1-2(\ln f)_{sy} = \rho^{-1}$.

Notice that the first equation in (\ref{CSP_BL2}) can be rewritten as
\begin{equation}
\left(\frac{g}{f} \right)_{sy} = \left(1-2(\ln f)_{sy} \right) \frac{g}{f}\,,
\end{equation}
or
\begin{equation}  \label{CSP_BL4}
\rho \left(\frac{g}{f} \right)_{sy} = \frac{g}{f}\,,
\end{equation}
which is converted into
\begin{equation}  \label{CSPE1}
\partial_x \left(-\partial_t - \frac 12 |q|^2 \partial_x \right)q = q\,,
\end{equation}
by using (\ref{CSP_BL3}). Eq. (\ref{CSPE1}) is nothing but the complex short
pulse equation (\ref{CSP}).
\end{proof}

$N$-soliton solution to the bilinear equations (\ref{CSPE_bilinear1})--(\ref%
{CSPE_bilinear2}) can be expressed by pfaffians similar to the ones for
coupled modified KdV equation \cite{IwaoHirota}. To this end, we need to
define two sets: $B_\mu$ ($\mu=1,2$): $B_1 = \{b_1, b_2, \cdots, b_{N} \}$, $%
B_2 = \{b_{N+1}, b_2, \cdots, b_{2N} \}$, and an index function of $b_j$ by $%
index(b_j) = \mu$ \ if $b_j \in B_\mu$.

\begin{theorem}
The pfaffians 
\begin{eqnarray}
f &=& \mathrm{Pf} (a_1, \cdots, a_{2N}, b_1, \cdots, b_{2N})\,,
\label{CSP_sl1} \\
g &=& \mathrm{Pf} (d_0, \beta_1, a_1, \cdots, a_{2N}, b_1, \cdots, b_{2N})\,.
\label{CSP_sl2}
\end{eqnarray}
satisfy the bilinear equations (\ref{CSPE_bilinear1})--(\ref{CSPE_bilinear2}%
) provided that the elements of the pfaffians are defined by
\begin{equation}  \label{CSPE_pf1}
\mathrm{Pf}(a_j,a_k)= \frac{p_j-p_k}{p_j+p_k} e^{\eta_j+\eta_k}\,, \quad
\mathrm{Pf}(a_j,b_k)=\delta_{j,k}\,,
\end{equation}
\begin{equation}  \label{CSPE_pf2}
\mathrm{Pf}(b_j,b_k)=\frac 14 \frac{\alpha_j \alpha_k}{p^{-2}_j-p^{-2}_{k}}
\delta_{\mu+1, \nu}\,, \quad \mathrm{Pf}(d_l,a_k)= p_k^{l} e^{\eta_k}\,,
\end{equation}

\begin{equation}  \label{CSPE_pf4}
\mathrm{Pf}(b_j,\beta_1)=\alpha_j \delta_{\mu, 1}\,, \quad \mathrm{Pf}%
(d_0,b_j) =\mathrm{Pf}(d_0,\beta_1) = \mathrm{Pf}(a_j,\beta_1)=0\,.
\end{equation}
Here $\mu=index(b_j)$, $\nu=index(b_k)$, $\eta_j=p_j y + p_j^{-1} s$ which
satisfying $p_{j+N}=\bar{p}_j$, $\alpha_{j+N}=\bar{\alpha}_{j}$, $\bar{p_j}$
and $\bar{\alpha}_{j}$ represent the complex conjugates of $p_j$ and ${\alpha%
}_{j}$, respectively. The same notation will be used hereafter.
\end{theorem}

The proof of the Theorem is given in Appendix. Combined with dependent and
hodograph transformations (\ref{CSP_dtrf})--(\ref{CSP_hodograph}), the above
pfaffians (\ref{CSP_sl1})--(\ref{CSP_sl2}) give $N$-soliton solution to the
complex short pulse equation (\ref{CSP}) in parametric form.

\subsection{One- and two-soliton solutions for the complex short pulse
equation}

In this subsection, we provide one- and two-soliton to the complex short
pulse equation (\ref{CSP}) and give a detailed analysis for their properties.

\subsubsection{One-soliton solution}

Based on (\ref{CSP_sl1})--(\ref{CSP_sl2}), the tau-functions for one-soliton
solution ($N=1$) are
\begin{eqnarray}
&& f = -1-\frac 14 \frac {|\alpha_1|^2(p_1\bar{p}_1)^2}{(p_1+\bar{p}_1)^2}
e^{\eta_1+\bar{\eta}_1} \,,
\end{eqnarray}
\begin{equation}
g = -\alpha_1 e^{\eta_1}\,.
\end{equation}

Let $p_1=p_{1R}+ \mathrm{i}p_{1I}$, and we assume $p_{1R} >0$ without loss
of generality, then the one-soliton solution can be expressed in the
following parametric form
\begin{equation}  \label{CSP1solitona}
q=\frac{\alpha_1}{|\alpha_1|}\frac{2p_{1R}}{|p_1|^2}e^{\mathrm{i} \eta_{1I}} %
\mbox{sech} \left(\eta_{1R}+\eta_{10}\right) \,,
\end{equation}
\begin{equation}  \label{CSP1solitonb}
x=y-\frac{2p_{1R}}{|p_1|^2}\left(\tanh
\left(\eta_{1R}+\eta_{10}\right)+1\right)\,, \quad t=-s\,,
\end{equation}
where
\begin{equation}
\eta_{1R}=p_{1R}y+\frac{p_{1R}}{|p_1|^2}s , \quad \eta_{1I}=p_{1I} y-\frac{%
p_{1I}}{|p_1|^2}s \,,\quad \eta_{10}=\ln \frac{|\alpha_1||p_1|^2}{4p_{1R}}\,.
\end{equation}

Eq. (\ref{CSP1solitona}) represents an envelope soliton of amplitude $%
2p_{1R}/|p_1|^2$ and phase $\eta_{1I}$. To analyze the property for the
one-soliton solution, we calculate out
\begin{equation}
\frac{\partial x}{\partial y} = 1- \frac{2p^2_{1R}}{|p_1|^2} {\mbox{sech}}%
^2(\eta_{1R}+\eta_{10})\,.
\end{equation}
Therefore, $\partial x/\partial y \to 1$ as $y \to \pm \infty$. Moreover, it
attains a minimum value of $({p^2_{1I}-p^2_{1R}})/({p^2_{1I}+p^2_{1R}})$
at the peak point of envelope soliton where $\eta_{1R}+\eta_{10}=0$. Since ${%
\partial |q|}/{\partial x}=\frac{\partial |q|/\partial y}{\partial
x/\partial y}$, we can classify this one-soliton solution as follows:

\begin{itemize}
\item \textbf{smooth soliton:} when $|p_{1R}| < |p_{1I}|$, ${\partial x}/{%
\partial y}$ is always positive, which leads to a smooth envelope soliton
similar to the envelope soliton for the nonlinear Schr{\"o}dinger equation.
An example with $p_1=1+1.5\mathrm{i}$ is illustrated in Fig. 1 (a).

\item \textbf{loop soliton:} when $|p_{1R}| > |p_{1I}|$, the minimum value
of ${\partial x}/{\partial y}$ at the peak point of the soliton becomes
negative. In view of the fact that $\partial x/\partial y \to 1$ as $y \to
\pm \infty$, ${\partial x}/{\partial y}$ has two zeros at both sides of the
peak of the envelope soliton. Moreover, ${\partial x}/{\partial y}< 0$
between these two zeros. This leads to a loop soliton for the envelope of $q$%
. An example is shown in Fig. (b) with $p_1=1+0.5\mathrm{i}$.

\item \textbf{cuspon soliton:} when $|p_{1R}| = |p_{1I}|$, ${\partial x}/{%
\partial y}$ has a minimum value of zero at $\eta_{1R}+\eta_{10}=0$, which
makes the derivative of the envelope $|q|$ with respect to $x$ going to
infinity at the peak point. Thus, we have a cusponed envelope soliton, which
is illustrated in Fig. 1 (c) with $p_1=1+\mathrm{i}$.
\end{itemize}
\begin{figure}[htbp]
\centerline{
\includegraphics[scale=0.35]{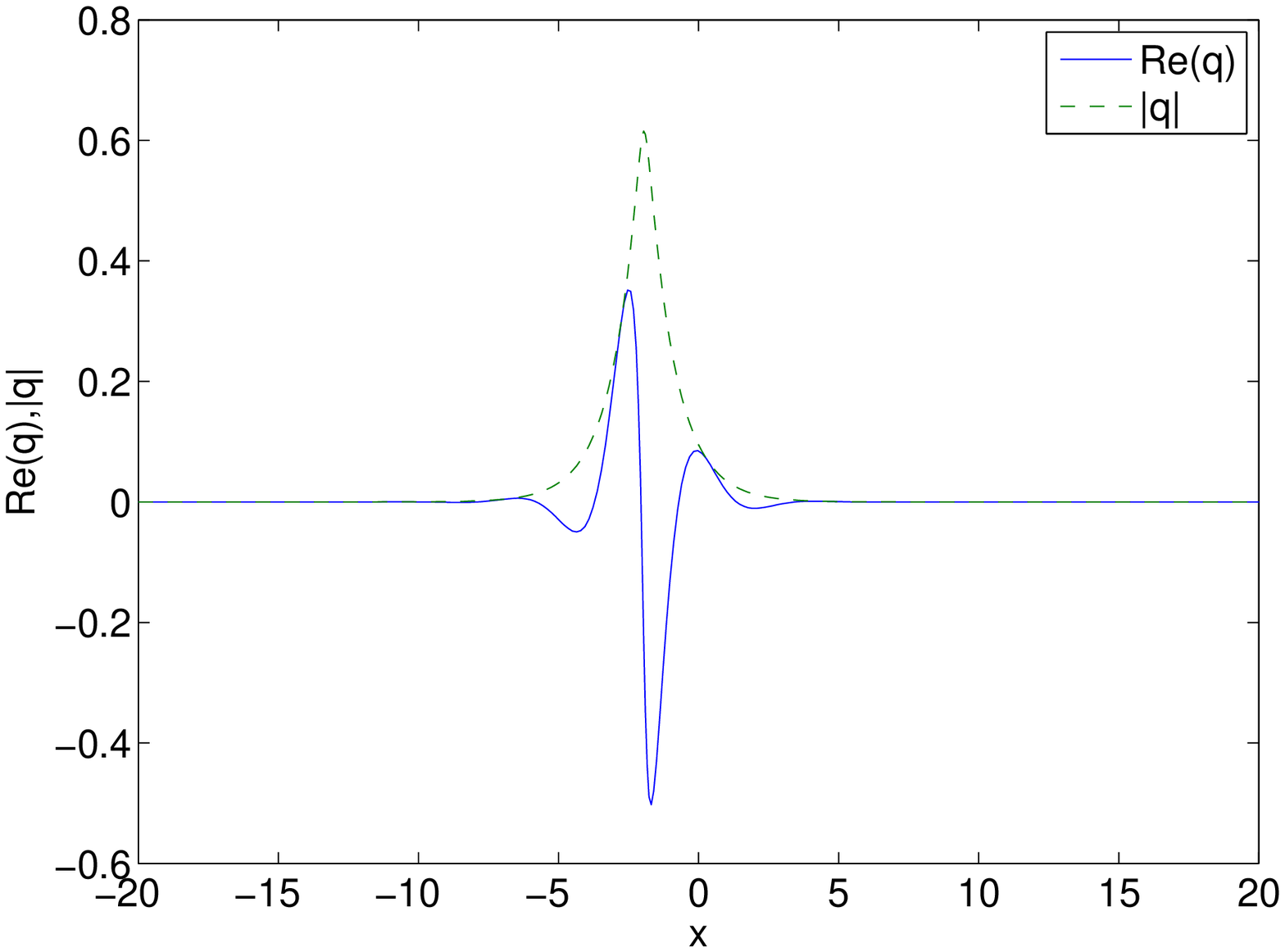}\quad
\includegraphics[scale=0.35]{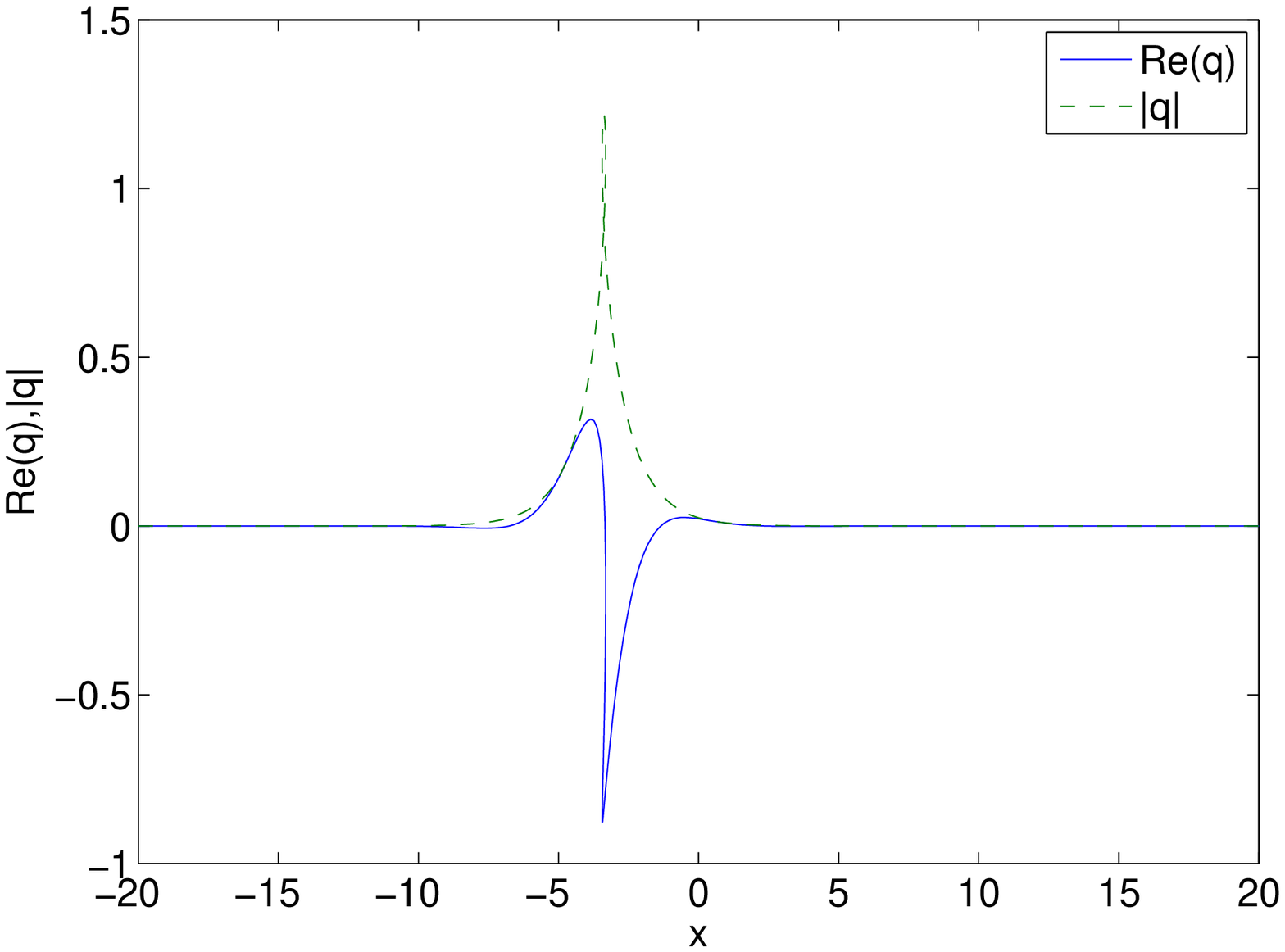}} \kern-0.3\textwidth
\hbox to
\textwidth{\hss(a)\kern16em\hss(b)\kern11em} \kern+0.3\textwidth
\centerline{
\includegraphics[scale=0.35]{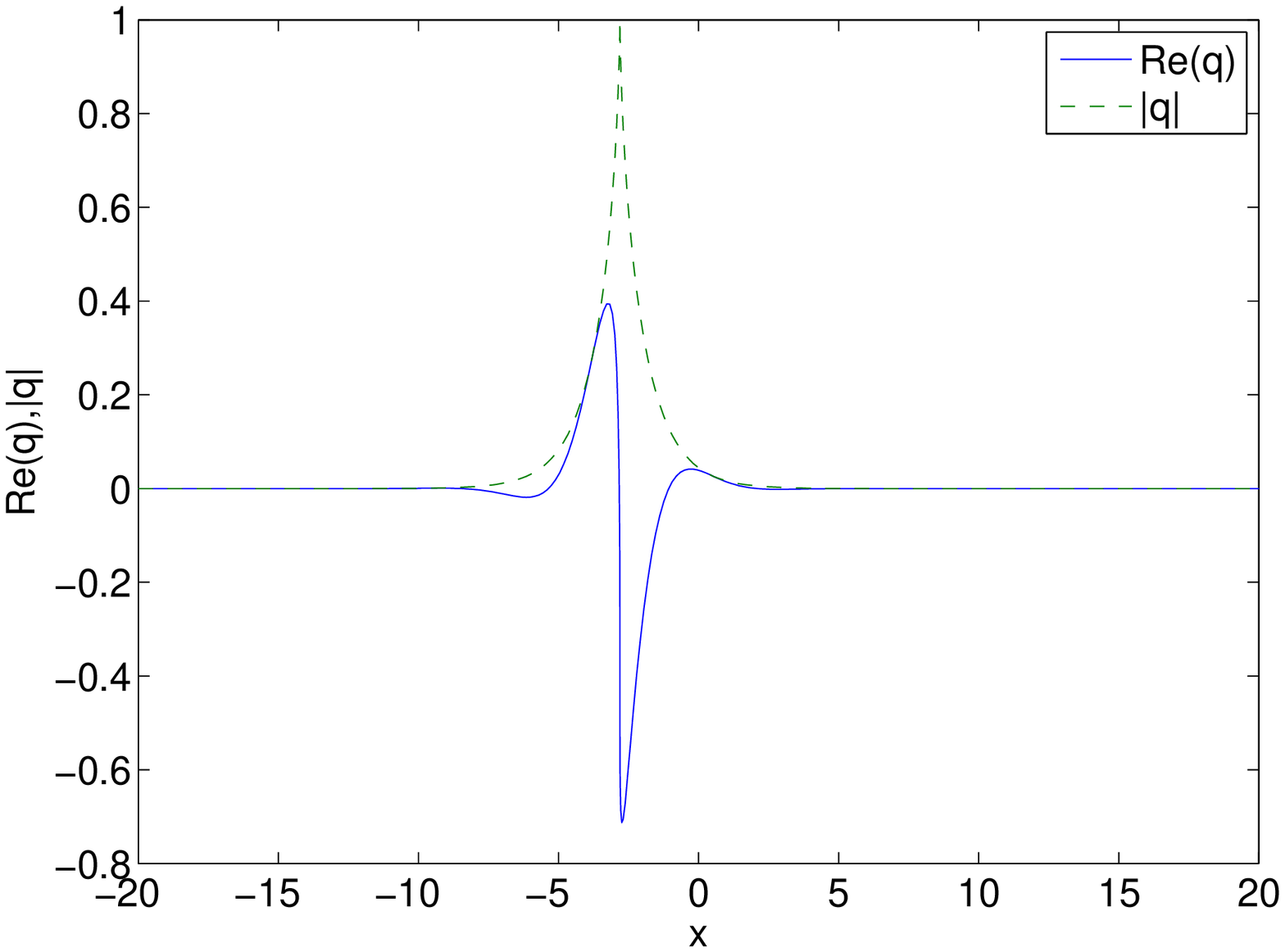}} \kern-0.3\textwidth
\hbox to
\textwidth{\hss(c)\kern21em} \kern+0.3\textwidth
\caption{Envelope soliton for the complex short pulse equation (\protect\ref%
{CSP}), solid line: $Re(q)$, dashed line: $|q|$; (a) smooth soliton with $%
p_1=1+1.5\mathrm{i}$, (b) loop soliotn with $p_1=1+0.5\mathrm{i}$, (c)
cuspon soliton with $p_1=1+\mathrm{i}$.}
\label{figure:cspe1soliton}
\end{figure}

\begin{remark}
The one-soliton solution to the short pulse equation (\ref{SPE}) is of
loop-type, which lacks physical meaning in the context of nonlinear optics.
However, the one-soliton solution to the complex short pulse equation (\ref%
{CSP}) is of breather-type, which allows physical meaning for optical pulse.
\end{remark}

\begin{remark}
When $|p_{1R}| <|p_{1I}|$, there is no singularity for one-soliton solution.
Moreover, in view of $\eta_{1R}$ associated with the width of envelope
soliton and $\eta_{1I}$ associated with the phase, it is obvious that this
nonsingular envelope soliton can only contain a few optical cycle. This
property coincides with the fact that the complex short pulse equation is
derived for the purpose of describing ultra-short pulse propagation. When $%
|p_{1R}| =|p_{1I}|$, the soliton becomes cuspon-like one, which agrees with
the results in \cite{Bandelow} derived from a bidirectional model.
\end{remark}

\subsubsection{Two-soliton solution}

Based on the $N$-soliton solution of the complex short pulse equation from (%
\ref{CSP_sl1})--(\ref{CSP_sl2}), the tau-functions for two-soliton solution
can be expanded for $N=2$
\begin{eqnarray}
&&f=\mathrm{Pf}(a_{1},a_{2},a_{3},a_{4},b_{1},b_{2},b_{3},b_{4})  \nonumber
\\
&&\quad =1+a_{1\bar{1}}e^{\eta _{1}+\bar{\eta}_{1}}+a_{1\bar{2}}e^{\eta _{1}+%
\bar{\eta}_{2}}+a_{2\bar{1}}e^{\eta _{2}+\bar{\eta}_{1}}+a_{2\bar{2}}e^{\eta
_{2}+\bar{\eta _{2}}}  \nonumber \\
&&\qquad +|P_{12}|^{2}\left( a_{1\bar{1}}a_{2\bar{2}}P_{1\bar{2}}P_{2\bar{1}%
}-a_{1\bar{2}}a_{2\bar{1}}P_{1\bar{1}}P_{2\bar{2}}\right) e^{\eta _{1}+\eta
_{2}+\bar{\eta}_{1}+\bar{\eta}_{2}}\,,
\end{eqnarray}

\begin{eqnarray}
&&g=\mathrm{Pf}(d_{0},\beta
_{1},a_{1},a_{2},a_{3},a_{4},b_{1},b_{2},b_{3},b_{4})  \nonumber \\
&&\quad =\alpha _{1}e^{\eta _{1}}+\alpha _{2}e^{\eta _{2}}+P_{12}\left(
\alpha _{1}P_{1\bar{1}}a_{2\bar{1}}-\alpha _{2}P_{2\bar{1}}a_{1\bar{1}%
}\right) e^{\eta _{1}+\eta _{2}+\bar{\eta}_{1}}  \nonumber \\
&&\qquad +P_{12}\left( \alpha _{1}P_{1\bar{2}}a_{2\bar{2}}-\alpha _{2}P_{2%
\bar{2}}a_{1\bar{2}}\right) e^{\eta _{1}+\eta _{2}+\bar{\eta}_{2}}\,,
\end{eqnarray}%
where
\begin{equation}
P_{ij}=\frac{p_{i}-p_{j}}{p_{i}+p_{j}}\,,\quad P_{i\bar{j}}=\frac{p_{i}-\bar{%
p}_{j}}{p_{i}+\bar{p}_{j}}\,,\quad a_{i\bar{j}}=\frac{\alpha _{i}\bar{\alpha}%
_{j}(p_{i}\bar{p}_{j})^{2}}{4(p_{i}+\bar{p}_{j})^{2}}\,,
\end{equation}%
and $\eta _{j}=p_{j}y+p_{j}^{-1}s$, $\bar{\eta}_{j}=\bar{p}_{j}y+\bar{p}%
_{j}^{-1}s$.

\begin{figure}[htbp]
\centerline{
\includegraphics[scale=0.35]{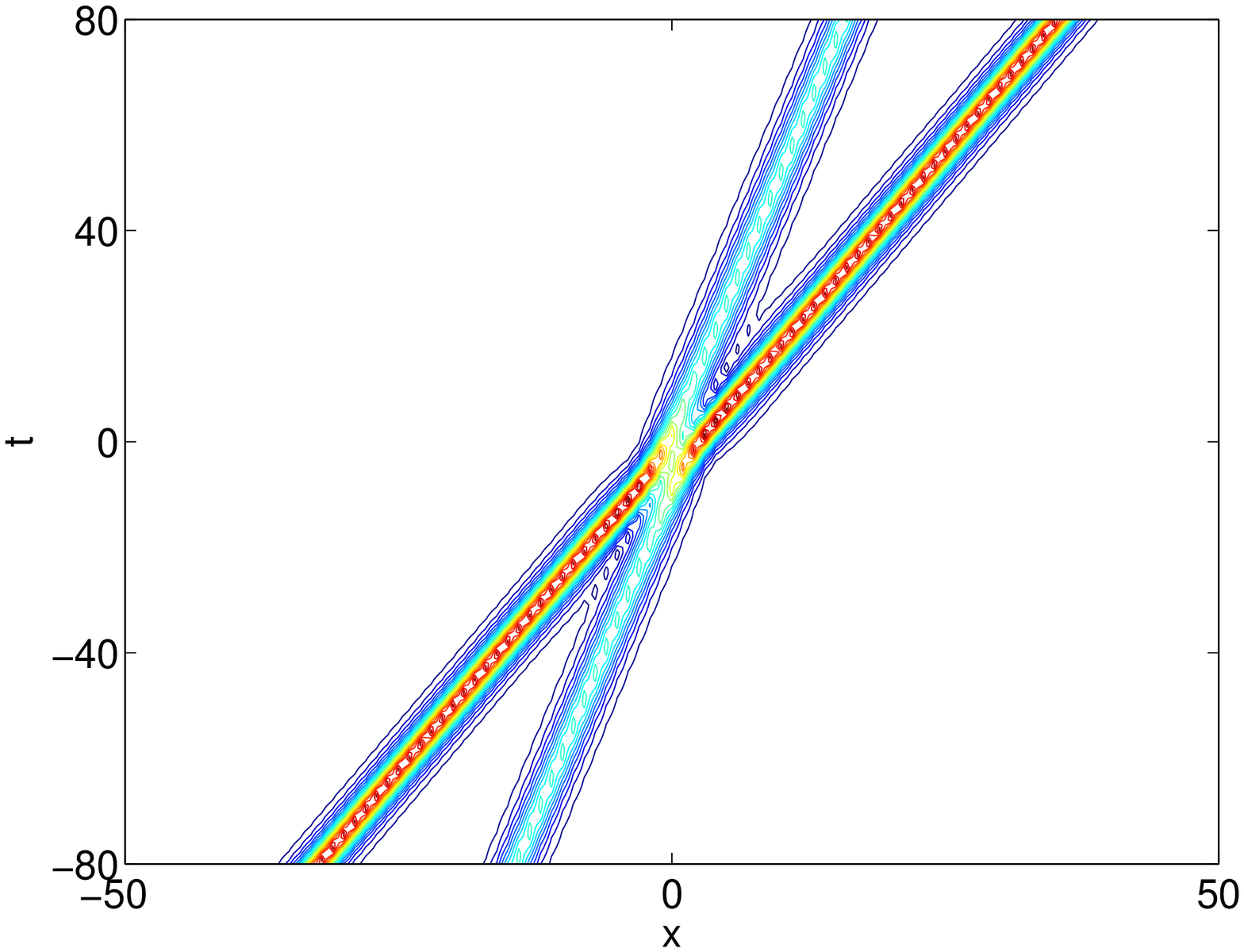}\quad
\includegraphics[scale=0.35]{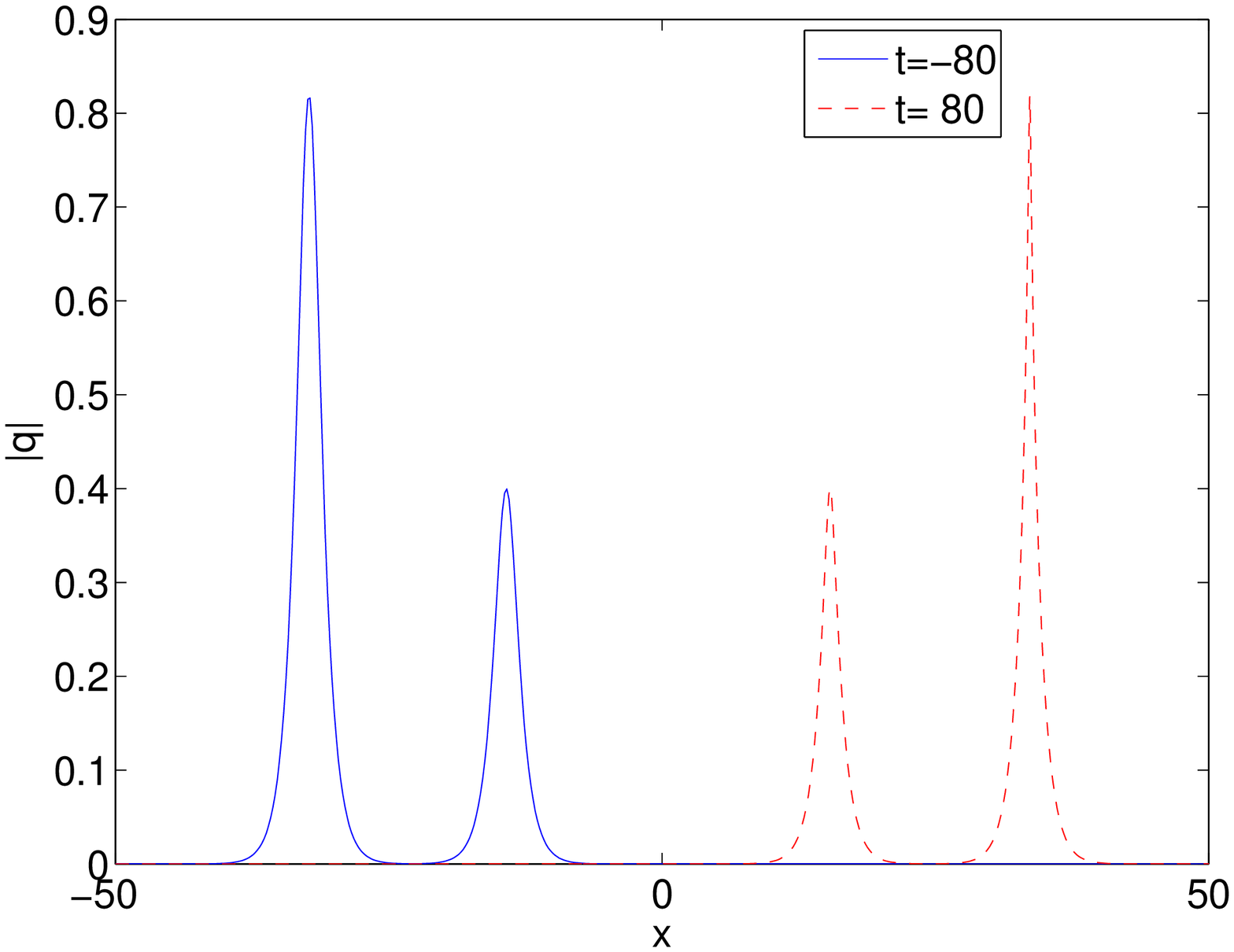}} \kern-0.315%
\textwidth \hbox to
\textwidth{\hss(a)\kern4em\hss(b)\kern2em} \kern+0.315\textwidth
\caption{Two-soliton solution to the complex short pulse equation (a)
contour plot; (b) profiles at $t=-80$, $80$.}
\label{f:1com2soliton}
\end{figure}
To avoid the singularity of the envelope solitons, the conditions $|p_{1R}|<
|p_{1I}|$ and $|p_{2R}|< |p_{2I}|$ need to be satisfied.
When two solitons stay apart, the amplitude of each soliton is of $2|p_{iR}|/|p_{i}|^2$, and the velocity is of $-1/|p_i|^2$ in the $ys$-coordinate system. Therefore, the soliton of larger velocity will catch up with and collide with the soliton of smaller velocity if it is initially located on the left. Furthermore, the
collision is elastic, and there is no change in shape and amplitude of
solitons except a phase shift. In Fig. 2, we illustrate
the contour plot for the collision of two solitons (a), as well as the
profiles (b) before and after the collision. The parameters are taken as $%
\alpha_1=\alpha_2=1.0$, $p_1=1+1.2\mathrm{i}$ and $p_2=1+2\mathrm{i}$.

Since the velocity of single envelope soliton is $-1/|p_i|^2$ in the $ys$%
-coordinate system, a bound state can be formed under the condition of $%
|p_1|^2 = |p_2|^2$ if two solitons stay close enough and move with the same
velocity. Such a bound state is shown in Fig. 3 for
parameters chosen as $\alpha_1=\alpha_2=1.0$, $p_1=1.3+1.8193\mathrm{i}$, $%
p_2=1+2\mathrm{i}$.
It is interesting that the envelope of the bound state oscillates
periodically as it moves along $x$-axis.
\begin{figure}[htbp]
\centerline{
\includegraphics[scale=0.35]{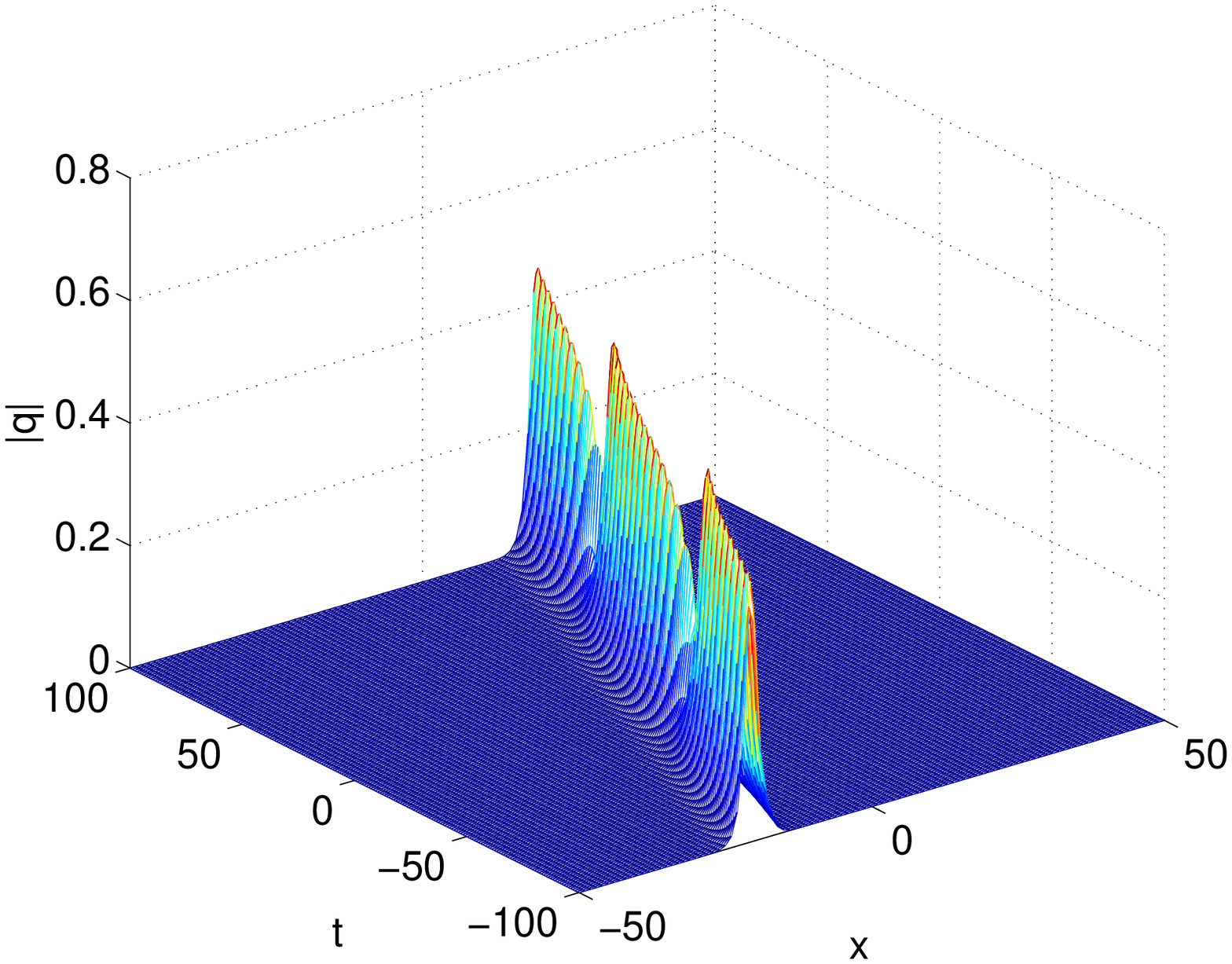}\quad
\includegraphics[scale=0.35]{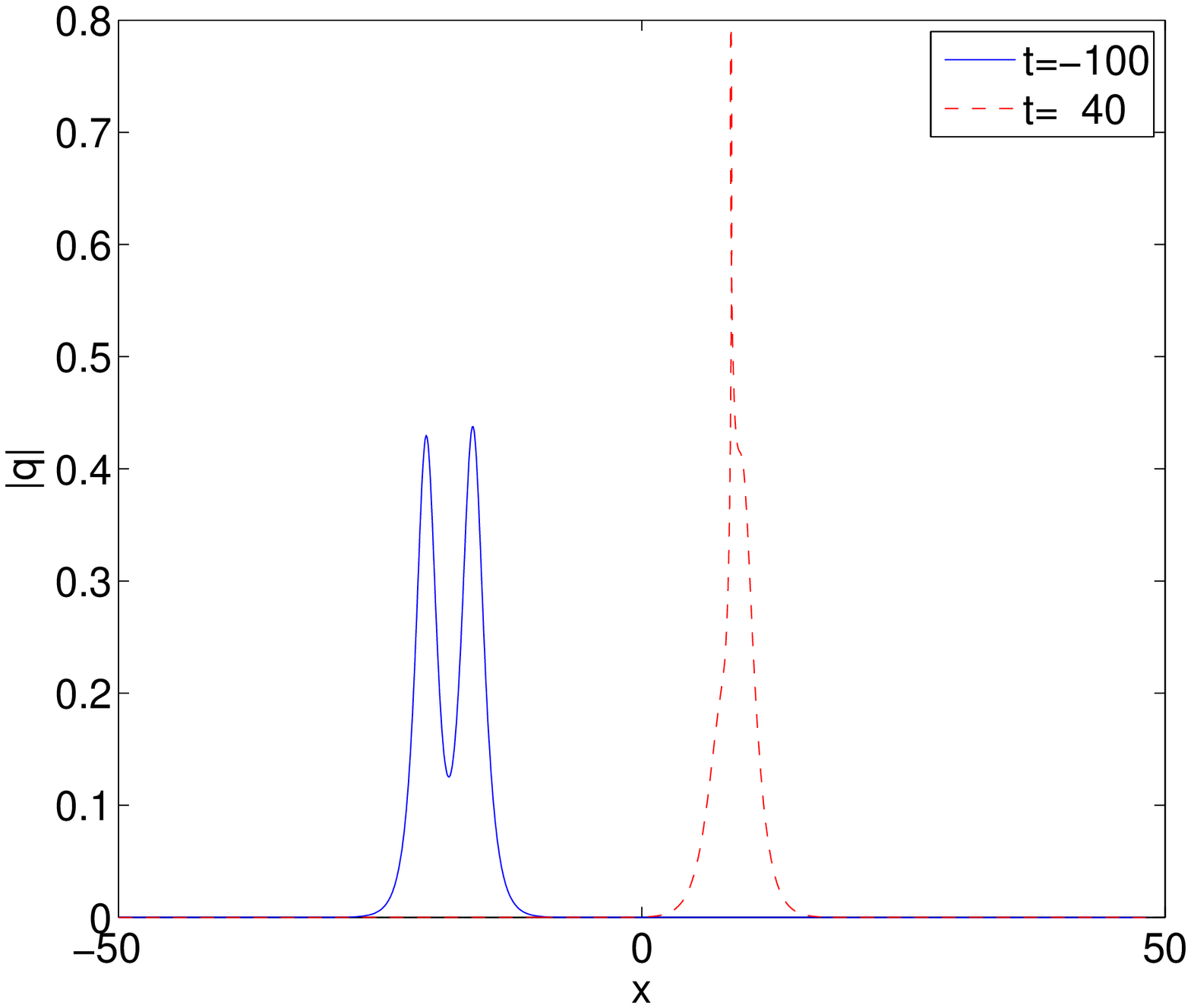}} \kern-0.315\textwidth
\hbox to
\textwidth{\hss(a)\kern4em\hss(b)\kern2em} \kern+0.315\textwidth
\caption{Bound state to the complex short pulse equation: (a) 3D plot (b)
profiles at $t=-100$, $40$.}
\label{f:1comboundstate}
\end{figure}
\subsection{Bilinear equations and $N$-soliton solutions to the coupled
complex short pulse equation}

\begin{proposition}
The coupled complex short pulse equation is derived from bilinear equations
\begin{equation}  \label{CCSPE_bilinear1}
D_sD_y f \cdot g_i =fg_i, \quad i=1,2 \,,
\end{equation}
\begin{equation}  \label{CCSPE_bilinear2}
D^2_s f \cdot f =\frac{1}{2} \left(|g_1|^2+|g_2|^2\right)\,,
\end{equation}
by dependent variable transformation
\begin{equation}  \label{CCSPE_vartrf}
q_1=\frac{g_1}{f}, \quad q_2=\frac{g_2}{f}\,,
\end{equation}
and hodograph transformation
\begin{equation}
x = y -2(\ln f)_s\,, \quad t=-s \,,  \label{CCSP_hodograph}
\end{equation}
\end{proposition}

\begin{proof}
Dividing both sides of Eqs. (\ref{CCSPE_bilinear1})--(\ref{CCSPE_bilinear2})
by $f^2$, we have
\begin{equation}
\left(\frac{g_i}{f} \right)_{sy} + 2\frac{g_i}{f} \left( \ln f\right)_{sy} =
\frac{g_i}{f}\,,  \label{CCSP_BL2a}
\end{equation}
\begin{equation}
\left( \ln f\right)_{ss} =\frac{1}{4} \left( \frac{|g_1|^2}{f^2}+\frac{%
|g_2|^2}{f^2} \right)\,.  \label{CCSP_BL2b}
\end{equation}

From dependent variable and hodograph transformations (\ref{CCSPE_vartrf})--(%
\ref{CCSP_hodograph}), we obtain
\[
\frac{\partial x}{\partial s} = -2(\ln f)_{ss} = -\frac 12
\left(|q_1|^2+|q_2|^2 \right)\,, \qquad \frac{\partial x}{\partial y} =
1-2(\ln f)_{sy}\,,
\]
which implies
\begin{equation}  \label{CCSP_BL3}
{\partial_y} = \rho^{-1} {\partial_x}\,, \qquad {\partial_s} = -{\partial_t}
- \frac 12 \left(|q_1|^2+|q_2|^2 \right) {\partial_x}\,
\end{equation}
by letting $1-2(\ln f)_{sy} = \rho^{-1}$.

With the use of (\ref{CCSP_BL3}), Eq. (\ref{CCSP_BL2a}) can be recast into
\begin{equation}  \label{CCSP_BL4}
\rho \left(\frac{g_i}{f} \right)_{sy} = \frac{g_i}{f}\,, \quad i=1,2\,,
\end{equation}
which can be further converted into
\begin{equation}  \label{CCSPE_alt}
\partial_x \left(-\partial_t - \frac 12 (|q_1|^2+|q_2|^2) \partial_x
\right)q_i = q_i\,, \quad i=1,2\,.
\end{equation}
Eq. (\ref{CCSPE_alt}) is, obviously, equivalent to the coupled complex short
pulse equation (\ref{CCSP1})--(\ref{CCSP2}).
\end{proof}


$N$-soliton solution for the coupled complex short pulse equation is given
in a similar way as the complex short pulse equation by the following
theorem. 


\begin{theorem}
The coupled complex short pulse equation admits the following $N$-soliton
solution
\[
q_i=\frac{g_i}{f}, \quad x = y -2(\ln f)_s\,, \quad t=-s \,,
\]
where $f$, $g_i$ are pfaffians defined as
\begin{eqnarray}  \label{CCSPE_Nsoliton1}
f &=& \mathrm{Pf} (a_1, \cdots, a_{2N}, b_1, \cdots, b_{2N})\,, \\
g_i &=& \mathrm{Pf} (d_0, \beta_{i}, a_1, \cdots, a_{2N}, b_1, \cdots,
b_{2N})\,,  \label{CCSPE_Nsoliton2}
\end{eqnarray}
and the elements of the pfaffians are determined as
\begin{equation}  \label{NCSPE_pf1}
\mathrm{Pf}(a_j,a_k)= \frac{p_j-p_k}{p_j+p_k} e^{\eta_j+\eta_k}\,, \quad
\mathrm{Pf}(a_j,b_k)=\delta_{j,k}\,,
\end{equation}
\begin{equation}  \label{NCSPE_pf2}
\mathrm{Pf}(b_j,b_k)=\frac 14 \frac{\sum^2_{i=1} \alpha^{(i)}_j
\alpha^{(i)}_k }{p^{-2}_j-p^{-2}_{k}} \delta_{\mu+1, \nu}\,, \quad \mathrm{Pf%
}(d_l,a_k)= p_k^{l} e^{\eta_k}\,,
\end{equation}

\begin{equation}  \label{NCSPE_pf4}
\mathrm{Pf}(b_j,\beta_i)=\alpha^{(i)}_j \delta_{\mu, i}\,,\quad \mathrm{Pf}%
(d_0,b_j) =\mathrm{Pf}(d_0,\beta_i) = \mathrm{Pf}(a_j,\beta_i)=0\,.
\end{equation}
Here $\mu=index(b_j)$, $\nu=index(b_k)$, $\eta_j=p_j y + p_j^{-1} s +
\eta_{j,0}$ which satisfying $p_{j+N}=\bar{p}_j$, $\alpha_{j+N}=\bar{\alpha}%
_{j}$.
\end{theorem}

The proof of the Theorem is given in the Appendix. In the subsequent
section, based on the $N$-soliton solution of coupled complex short pulse
equation, we will investigate the dynamics of one- and two-solitons in
details.

\begin{remark}
Through the transformations
\begin{equation}  \label{NCSPE_trfs}
x = y -2(\ln f)_s\,, \quad t=-s \,, \quad q_i=\frac{g_i}{f}\,,
\end{equation}
the vector complex short pulse equation (\ref{NCSPE}) can be decomposed into
the following bilinear equations
\begin{equation}  \label{NCSPE_bilinear1}
D_sD_y f \cdot g_i =fg_i, \quad i=1, \cdots, n \,,
\end{equation}
\begin{equation}  \label{NCSPE_bilinear2}
D^2_s f \cdot f =\frac{1}{2} \left(\sum^n_{i=1}|g_i|^2\right)\,.
\end{equation}
The parametric form of $N$-soliton solution in terms of pfaffians to the
vector complex short pulse equation (\ref{NCSPE}) can be given in a very
similar from as to to the coupled complex short pulse equation. Here, we
omit the details and will report the results later on.
\end{remark}

\section{Dynamics of solitons to the coupled complex short pulse equation}
\subsection{One-soliton solution}
The tau-functions for one-soliton solution to the coupled complex short
pulse equation are obtained from (\ref{CCSPE_Nsoliton1})--(\ref%
{CCSPE_Nsoliton2}) for $N=1$
\begin{equation}
f = -1-\frac 14 \frac {\sum_{i=1}^2|\alpha^{(i)}_1|^2(p_1\bar{p}_1)^2}{(p_1+%
\bar{p}_1)^2} e^{\eta_1+\bar{\eta}_1} \,,
\end{equation}
\begin{equation}
g_1 = -\alpha^{(1)}_1 e^{\eta_1}\,, \quad g_2 = -\alpha^{(2)}_1 e^{\eta_1}\,.
\end{equation}

Let $p_{1}=p_{1R}+\mathrm{i}p_{1I}$, the one-soliton solution can be
expressed in the following parametric form
\begin{equation}
\left(
\begin{array}{c}
q_{1} \\
q_{2}%
\end{array}%
\right) =\left(
\begin{array}{c}
A_{1} \\
A_{2}%
\end{array}%
\right) \frac{2p_{1R}}{|p_{1}|^{2}}e^{\mathrm{i}\eta _{1I}}{\mbox{sech}}%
\left( \eta _{1R}+\eta _{10}\right) \,,  \label{1soliton_ay}
\end{equation}%
\begin{equation}
x=y-\frac{2p_{1R}}{|p_{1}|^{2}}\left( \tanh (\eta _{1R}+\eta _{10})+1\right)
\,,\quad t=-s\,,  \label{CCSP1solitonb}
\end{equation}%
where
\begin{equation}
\eta _{1R}=p_{1R}\left( y+\frac{1}{|p_{1}|^{2}}s\right) ,\quad \eta
_{1I}=p_{1I}\left( y-\frac{1}{|p_{1}|^{2}}s\right) \,,
\end{equation}%
\begin{equation}
A_{i}=\frac{\alpha _{1}^{(i)}}{\sqrt{\sum_{i=1}^{2}|\alpha _{1}^{(i)}|^2}}%
\,,\quad \eta_{10}=\ln \frac{\sqrt{\sum_{i=1}^{2}|\alpha _{1}^{(i)}|^2}%
|p_{1}|^{2}}{4|p_{1R}|}\,.
\end{equation}%
The amplitudes of the single soliton in each component are ${2|A_{1}|p_{1R}}/%
{|p_{1}|^{2}}$ and ${2|A_{2}|p_{1R}}/{|p_{1}|^{2}}$, respectively. Note that
$|A_{1}|^{2}+|A_{2}|^{2}=1$. Same as the analysis for one-soliton solution
of complex short pulse equation, if $|p_{1R}|<|p_{1I}|$, the envelope for
one-soliton in each of the component is smooth, whereas, if $%
|p_{1R}|>|p_{1I}| $, it becomes a loop (multi-valued) soliton, if $%
|p_{1R}|=|p_{1I}|$, it is a cuspon.

\subsection{Soliton interactions}
Two-soliton solution for coupled complex short pulse equation is obtained
from (\ref{CCSPE_Nsoliton1})--(\ref{CCSPE_Nsoliton2}) for $N=2$. By
expanding the pfaffians, the tau-functions for two-soliton solution are
expressed by
\begin{eqnarray}
&&f=1+e^{\eta _{1}+\bar{\eta}_{1}+r_{1\bar{1}}}+e^{\eta _{1}+\bar{\eta}%
_{2}+r_{1\bar{2}}}+e^{\eta _{2}+\bar{\eta}_{1}+r_{2\bar{1}}}+e^{\eta _{2}+%
\bar{\eta}_{2}+r_{2\bar{2}}}  \nonumber \\
&&\qquad +|P_{12}|^{2}|P_{1\bar{2}}|^{2}P_{1\bar{1}}P_{2\bar{2}}\left( B_{1%
\bar{1}}B_{2\bar{2}}-B_{2\bar{1}}B_{1\bar{2}}\right) e^{\eta _{1}+\eta _{2}+%
\bar{\eta}_{1}+\bar{\eta}_{2}}\,,
\end{eqnarray}

\begin{eqnarray}
&& g_1= \alpha^{(1)}_{1} e^{\eta_1} + \alpha^{(1)}_2 e^{\eta_2} +P_{12} P_{1%
\bar{1}} P_{2\bar{1}} \left( \alpha^{(1)}_2 B_{1\bar{1}} - \alpha^{(1)}_1
B_{2\bar{1}} \right) e^{\eta_1+\eta_2+\bar{\eta}_1}  \nonumber \\
&& \qquad +P_{12} P_{1\bar{2}} P_{2\bar{2}} \left( \alpha^{(1)}_2 B_{1\bar{2}%
} - \alpha^{(1)}_1 B_{2\bar{2}} \right) e^{\eta_1+\eta_2+\bar{\eta}_2} \,,
\end{eqnarray}

\begin{eqnarray}
&&g_{2}=\alpha _{1}^{(2)}e^{\eta _{1}}+\alpha _{2}^{(2)}e^{\eta
_{2}}+P_{12}P_{1\bar{1}}P_{2\bar{1}}\left( \alpha _{2}^{(2)}B_{1\bar{1}%
}-\alpha _{1}^{(2)}B_{2\bar{1}}\right) e^{\eta _{1}+\eta _{2}+\bar{\eta}_{1}}
\nonumber \\
&&\qquad +P_{12}P_{1\bar{2}}P_{2\bar{2}}\left( \alpha _{2}^{(2)}B_{1\bar{2}%
}-\alpha _{1}^{(2)}B_{2\bar{2}}\right) e^{\eta _{1}+\eta _{2}+\bar{\eta}%
_{2}}\,,
\end{eqnarray}%
where
\[
P_{ij}=\frac{p_{i}-p_{j}}{p_{i}+p_{j}}\,,\quad P_{i\bar{j}}=\frac{p_{i}-\bar{%
p}_{j}}{p_{i}+\bar{p}_{j}}\,,
\]%
\[
B_{i\bar{j}}=\frac{\alpha _{i}^{(1)}\bar{\alpha}_{j}^{(1)}+\alpha _{i}^{(2)}%
\bar{\alpha}_{j}^{(2)}}{4(p_{i}^{-2}-\bar{p}_{j}^{-2})}\,,\quad e^{r_{i\bar{j%
}}}=\frac{\alpha _{i}^{(1)}\bar{\alpha}_{j}^{(1)}+\alpha _{i}^{(2)}\bar{%
\alpha}_{j}^{(2)}}{4(p_{i}^{-1}+\bar{p}_{j}^{-1})^2}\,.
\]%
and $\eta _{j}=p_{j}y+p_{j}^{-1}s$, $p_{3}=\bar{p}_{1}$, $p_{4}=\bar{p}_{2}$%
, thus, $\eta _{3}=\bar{\eta}_{1}$, $\eta _{4}=\bar{\eta}_{2}$.

Next, we investigate the asymptotic behavior of two-soliton solution. To
this end, we assume $p_{1R}>p_{2R}>0$, $p_{1R}/|p_{1}|^{2}>p_{2R}/|p_{2}|^{2}
$ without loss of generality. For the above choice of parameters, we have
(i) $\eta _{1R}\approx 0$, $\eta _{2R}\rightarrow \mp \infty $ as $%
t\rightarrow \mp \infty $ for soliton 1 and (ii) $\eta _{2R}\approx 0$, $%
\eta _{2R}\rightarrow \pm \infty $ as $t\rightarrow \mp \infty $ for soliton
2. This leads to the following asymptotic forms for two-soliton solution.
\newline
(i) Before collision ($t\rightarrow -\infty $) \newline
Soliton 1 ($\eta _{1R}\approx 0$, $\eta _{2R}\rightarrow -\infty $):
\begin{eqnarray}
\left(
\begin{array}{c}
q_{1} \\
q_{2}%
\end{array}%
\right)  &\rightarrow &\left(
\begin{array}{c}
\alpha _{1}^{(1)} \\
\alpha _{1}^{(2)}%
\end{array}%
\right) \frac{e^{\eta _{1}}}{1+e^{\eta _{1}+\bar{\eta}_{1}+r_{1\bar{1}}}}\,,
\nonumber  \label{soliton1_aybf} \\
&\rightarrow &\left(
\begin{array}{c}
A_{1}^{1-} \\
A_{2}^{1-}%
\end{array}%
\right) \frac{2p_{1R}}{|p_{1}|^{2}}e^{i\eta _{1I}}{\mbox{sech}}\left( \eta
_{1R}+\frac{r_{1\bar{1}}}{2}\right) \,,
\end{eqnarray}%
where
\begin{equation}
\left(
\begin{array}{c}
A_{1}^{1-} \\
A_{2}^{1-}%
\end{array}%
\right) =\left(
\begin{array}{c}
\alpha _{1}^{(1)} \\
\alpha _{1}^{(2)}%
\end{array}%
\right) \frac{1}{\sqrt{|\alpha _{1}^{(1)}|^{2}+|\alpha _{1}^{(2)}|^{2}}}\,.
\end{equation}

Soliton 2 ($\eta_{2R} \approx 0$, $\eta_{1R} \to \infty$):
\begin{equation}  \label{soliton2_aybf}
\left(%
\begin{array}{c}
q_1 \\
q_2%
\end{array}%
\right) \to \left(%
\begin{array}{c}
A^{2-}_1 \\
A^{2-}_2%
\end{array}%
\right)\frac{2p_{2R}}{|p_{2}|^2} e^{i\eta_{2I}} {\mbox{sech}}
\left(\eta_{2R}+\frac {r_{1\bar{1}2\bar{2}}-r_{1\bar{1}}}{2} \right)\,,
\end{equation}
where
\begin{equation}
\left(%
\begin{array}{c}
A^{2-}_1 \\
A^{2-}_2%
\end{array}%
\right) = \left(%
\begin{array}{c}
e^{r^{(1)}_{1\bar{1}2}} \\
e^{r^{(2)}_{1\bar{1}2}}%
\end{array}%
\right) \frac{e^{-(r_{1\bar{1}2\bar{2}}+r_{1\bar{1}}-r_{2\bar{2}})/{2}}} {%
\sqrt{|\alpha^{(1)}_2|^2+|\alpha^{(2)}_2|^2}} \,,
\end{equation}
with
\begin{equation}
e^{r^{(i)}_{1\bar{1}2}}= P_{12} P_{1\bar{1}} P_{2\bar{1}} \left(
\alpha^{(i)}_2 B_{1\bar{1}} - \alpha^{(i)}_1 B_{2\bar{1}} \right)\,, \quad
(i=1,2)
\end{equation}
\begin{equation}
e^{r_{1\bar{1}2\bar{2}}}=|P_{12}|^2 |P_{1\bar{2}}|^2 P_{1\bar{1}} P_{2\bar{2}%
} \left(B_{1\bar{1}} B_{2\bar{2}} -B_{2\bar{1}}B_{1\bar{2}}\right)\,.
\end{equation}
\newline
After collision ($t \to \infty$) \newline
Soliton 1 ($\eta_{1R} \approx 0$, $\eta_{2R} \to \infty$):

\begin{equation}
\left(
\begin{array}{c}
q_{1} \\
q_{2}%
\end{array}%
\right) \rightarrow \left(
\begin{array}{c}
A_{1}^{1+} \\
A_{2}^{1+}%
\end{array}%
\right) \frac{2p_{1R}}{|p_{1}|^{2}}e^{i\eta _{1I}}{\mbox{sech}}\left( \eta
_{2R}+\frac{r_{1\bar{1}2\bar{2}}-r_{2\bar{2}}}{2}\right) \,,
\label{soliton2_ayafter}
\end{equation}%
where
\begin{equation}
\left(
\begin{array}{c}
A_{1}^{1+} \\
A_{2}^{1+}%
\end{array}%
\right) =\left(
\begin{array}{c}
e^{r_{12\bar{1}}^{(1)}} \\
e^{r_{12\bar{1}}^{(2)}}%
\end{array}%
\right) \frac{e^{-(r_{1\bar{1}2\bar{2}}-r_{1\bar{1}}+r_{2\bar{2}})/{2}}}{%
\sqrt{|\alpha _{1}^{(1)}|^{2}+|\alpha _{1}^{(2)}|^{2}}}\,,
\end{equation}%
with
\begin{equation}
e^{r_{12\bar{1}}^{(i)}}=P_{12}P_{1\bar{2}}P_{2\bar{2}}\left( \alpha
_{2}^{(i)}B_{1\bar{2}}-\alpha _{1}^{(i)}B_{2\bar{2}}\right) \,,\quad
(i=1,2)\,.
\end{equation}

Soliton 2 ($\eta _{2R}\approx 0$, $\eta _{1R}\rightarrow -\infty $):
\begin{equation}
\left(
\begin{array}{c}
q_{1} \\
q_{2}%
\end{array}%
\right) \rightarrow \left(
\begin{array}{c}
A_{1}^{2+} \\
A_{2}^{2+}%
\end{array}%
\right) \frac{2p_{2R}}{|p_{2}|^{2}}e^{i\eta _{2I}}{\mbox{sech}}\left( \eta
_{2R}+\frac{r_{2\bar{2}}}{2}\right) \,,  \label{soliton2_ayaf}
\end{equation}%
where
\begin{equation}
\left(
\begin{array}{c}
A_{1}^{2+} \\
A_{2}^{2+}%
\end{array}%
\right) =\left(
\begin{array}{c}
\alpha _{2}^{(1)} \\
\alpha _{2}^{(2)}%
\end{array}%
\right) \frac{1}{\sqrt{|\alpha _{2}^{(1)}|^{2}+|\alpha _{2}^{(2)}|^{2}}}\,.
\end{equation}

Similar to the analysis for the CNLS equations \cite%
{Lakshmanan1997,Lakshmanan2001,Lakshmanan2003}, the change in the amplitude
of each of the solitons in each component can be obtained by introducing the
transition matrix $T^k_j$ by $A_j^{k+}= T_j^k A_j^{k-}$, $j,k=1,2$. The
elements of transition matrix is obtained from the above asymptotic analysis
as
\begin{equation}  \label{trasition_matrix1}
T_j^1=\left(\frac{P_{12}P_{1\bar{2}}}{\bar{P}_{12}\bar{P}_{1\bar{2}}}
\right)^{1/2} \frac{1}{\sqrt{1-\lambda_1 \lambda_2}} \left(1-\lambda_2\frac{%
\alpha_2^{(j)}}{\alpha_1^{(j)}} \right)\,, \quad j=1,2\,,
\end{equation}
\begin{equation}  \label{trasition_matrix2}
T_j^2=\left(\frac{\bar{P}_{12}P_{1\bar{2}}}{P_{12}\bar{P}_{1\bar{2}}}
\right)^{1/2} \sqrt{1-\lambda_1 \lambda_2} \left(1-\lambda_1\frac{%
\alpha_1^{(j)}}{\alpha_2^{(j)}} \right)^{-1}\,, \quad j=1,2\,,
\end{equation}
where $\lambda_1=B_{2\bar{1}}/B_{1\bar{1}}$, $\lambda_2=B_{1\bar{2}}/B_{2%
\bar{2}}$.

Therefore, in general, there is an exchange of energies between two components of two solitons
after the collision.
An example is shown in Fig. 4 for the
parameters taken as follows 
$p_{1}=1+1.2\mathrm{i}$, $p_{2}=1+2\mathrm{i}$, $\alpha^{(1)}_{1}=%
\alpha^{(2)}_{1}=1.0$, $\alpha^{(1)} _{2}=2.0$, $\alpha^{(2)}_{2}=1.0$.
\begin{figure}[tbph]
\centerline{
\includegraphics[scale=0.35]{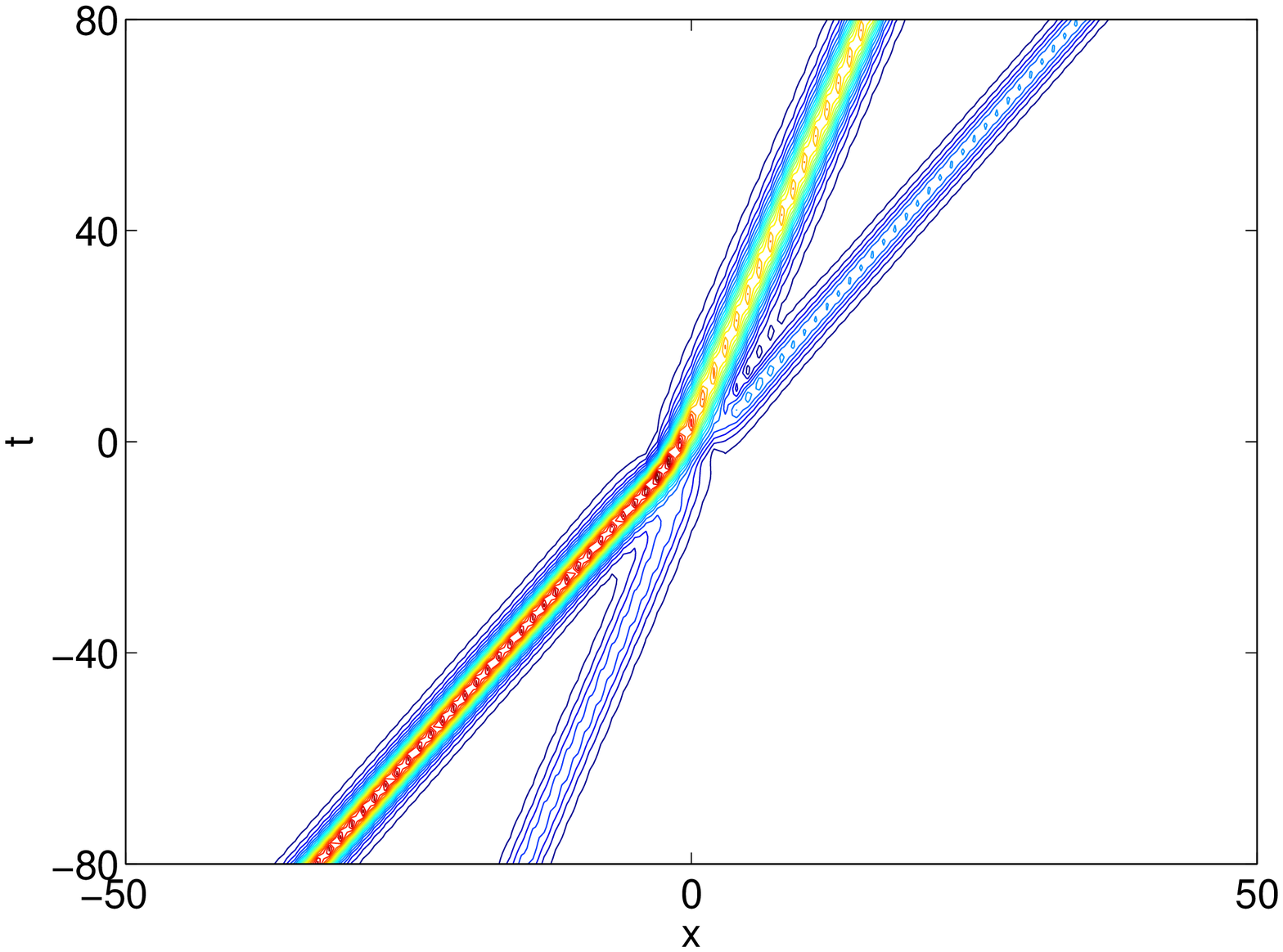}\quad
\includegraphics[scale=0.35]{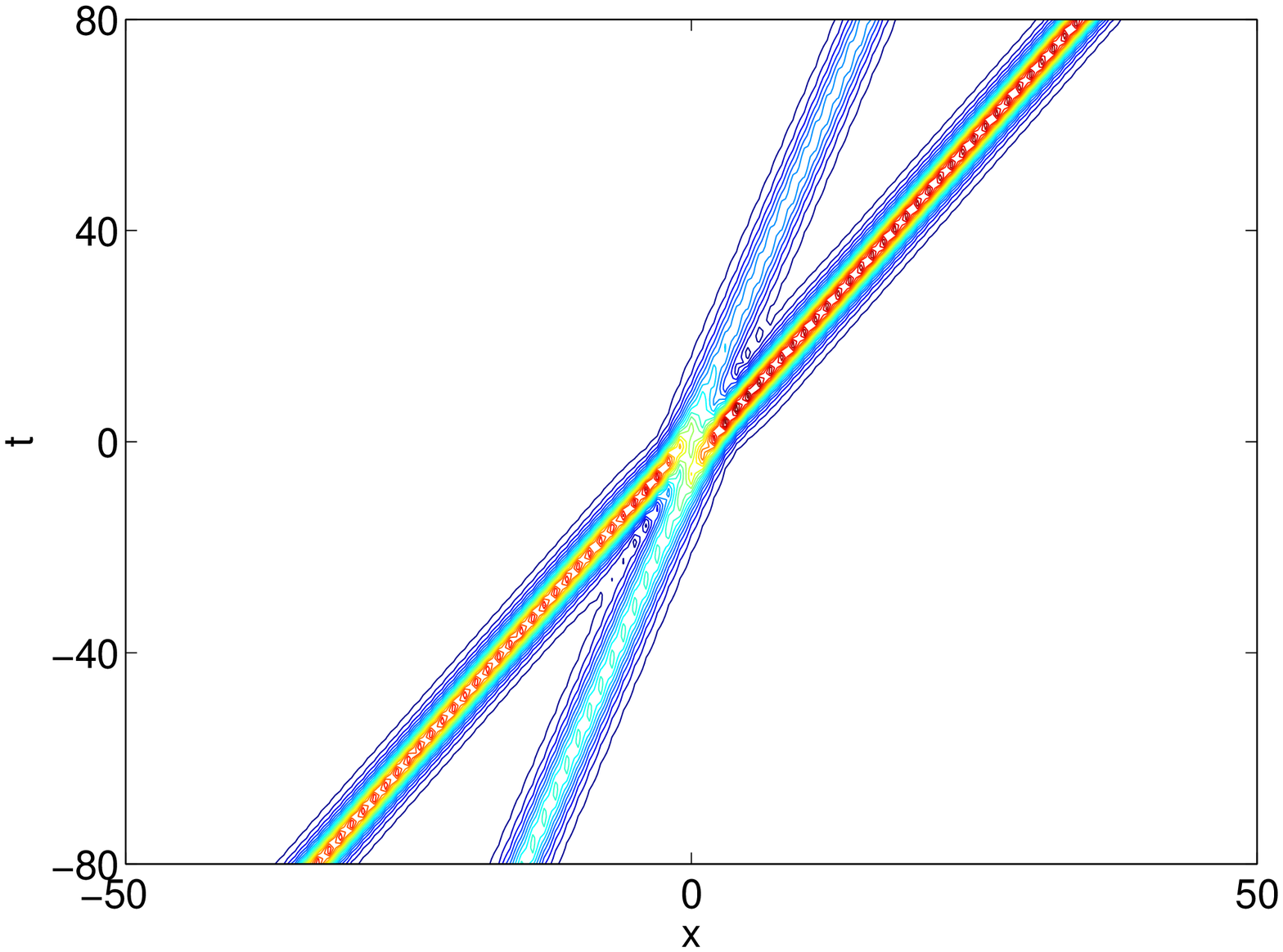}} \kern-0.315\textwidth
\hbox to
\textwidth{\hss(a)\kern6.5em\hss(b)\kern-1.5em} \kern+0.315\textwidth
\centerline{
\includegraphics[scale=0.35]{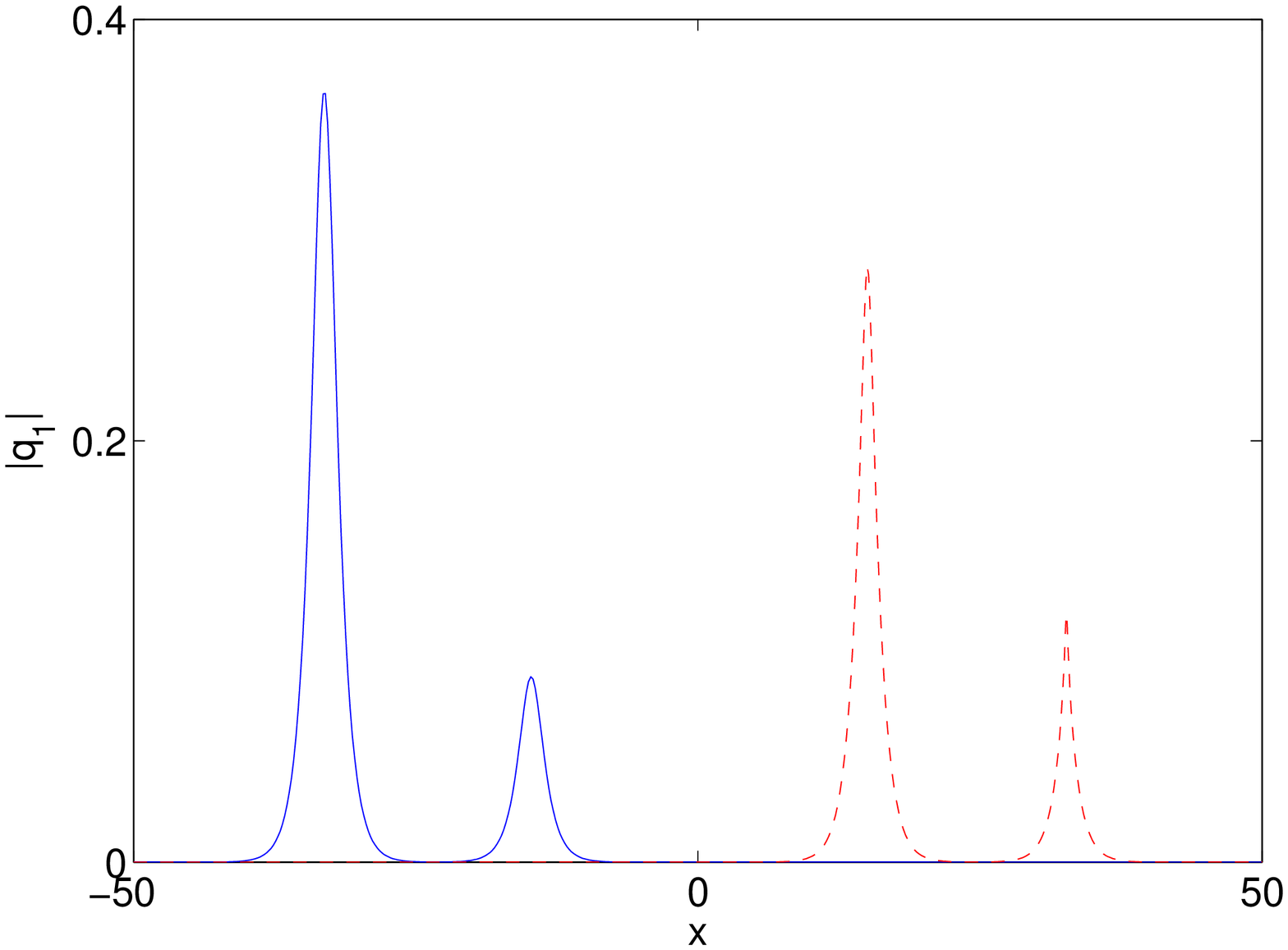}\quad
\includegraphics[scale=0.35]{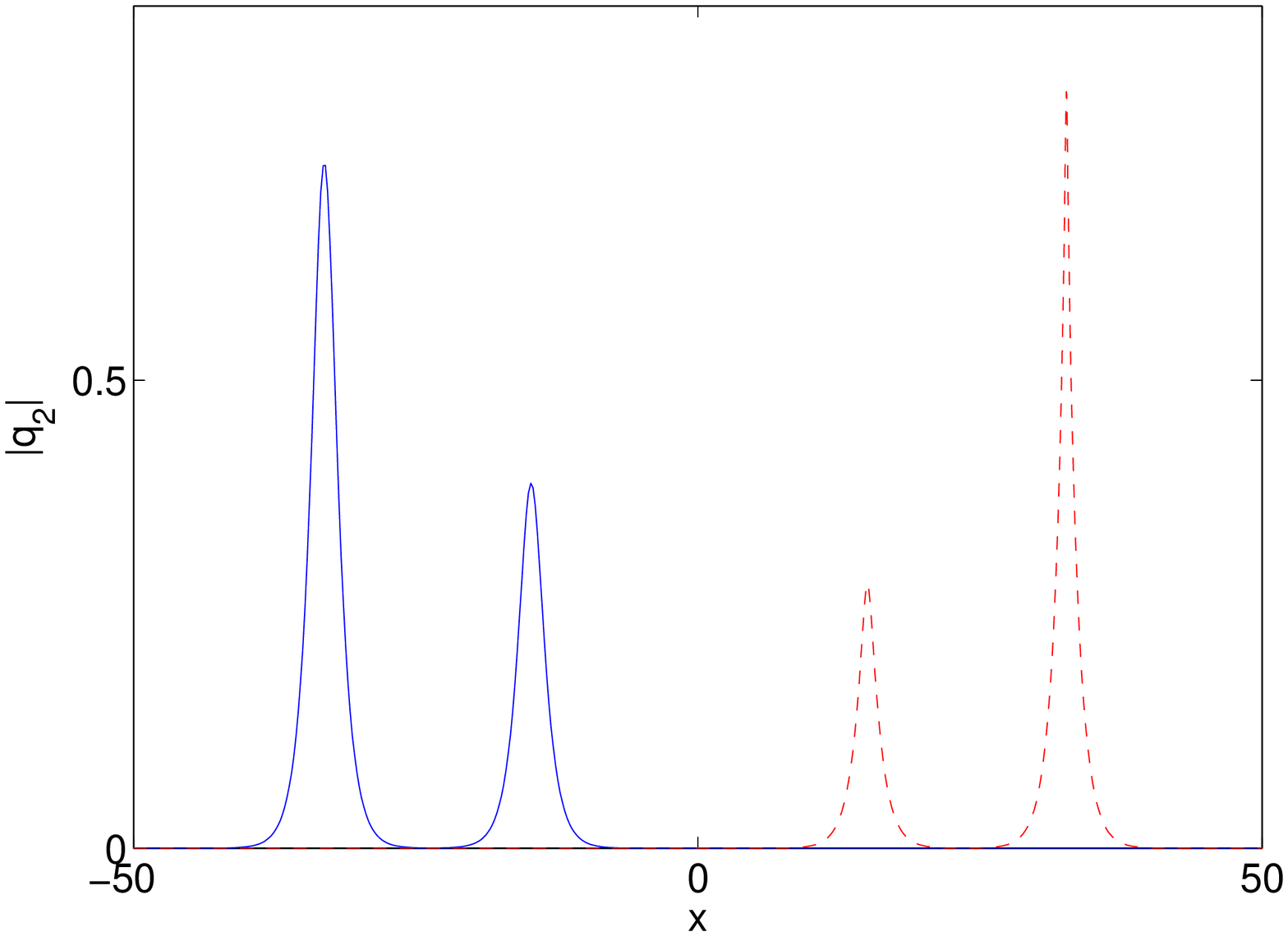}} \kern-0.315\textwidth
\hbox to
\textwidth{\hss(c)\kern6.5em\hss(d)\kern-1.5em} \kern+0.315\textwidth
\caption{Inelastic collision in coupled complex short
pulse equation. (a)-(b): contour plot; (c)-(d): profiles before and after the collision.}
\label{f:inelastic1}
\end{figure}
However, only for the special case
\begin{equation}
\frac{\alpha _{1}^{(1)}}{\alpha _{2}^{(1)}}=\frac{\alpha _{1}^{(2)}}{\alpha
_{2}^{(2)}}\,,
\end{equation}%
there is no energy exchange between two compoents of solitons
after the collision.  An example is shown in Fig. 5 for the
parameters 
$p_{1}=1+1.2\mathrm{i}$,\ $p_{2}=1+2\mathrm{i}$,\ $\alpha^{(1)}_{1}=%
\alpha^{(2)}_{1}=1.0$, \ $\alpha^{(1)}_{2}=\alpha^{(2)}_{2}=1.0$.
\begin{figure}[htbp]
\centerline{
\includegraphics[scale=0.35]{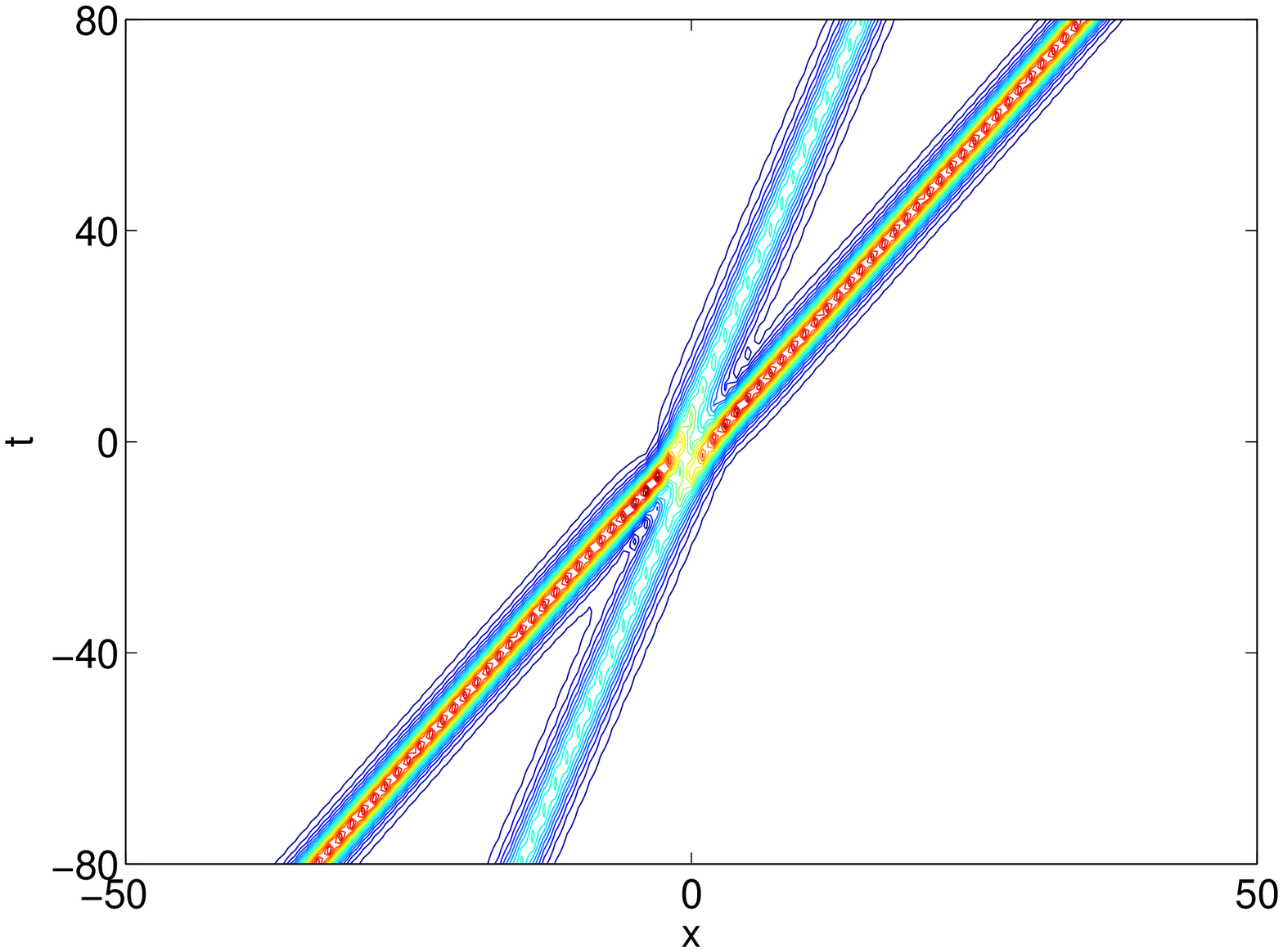}\quad
\includegraphics[scale=0.35]{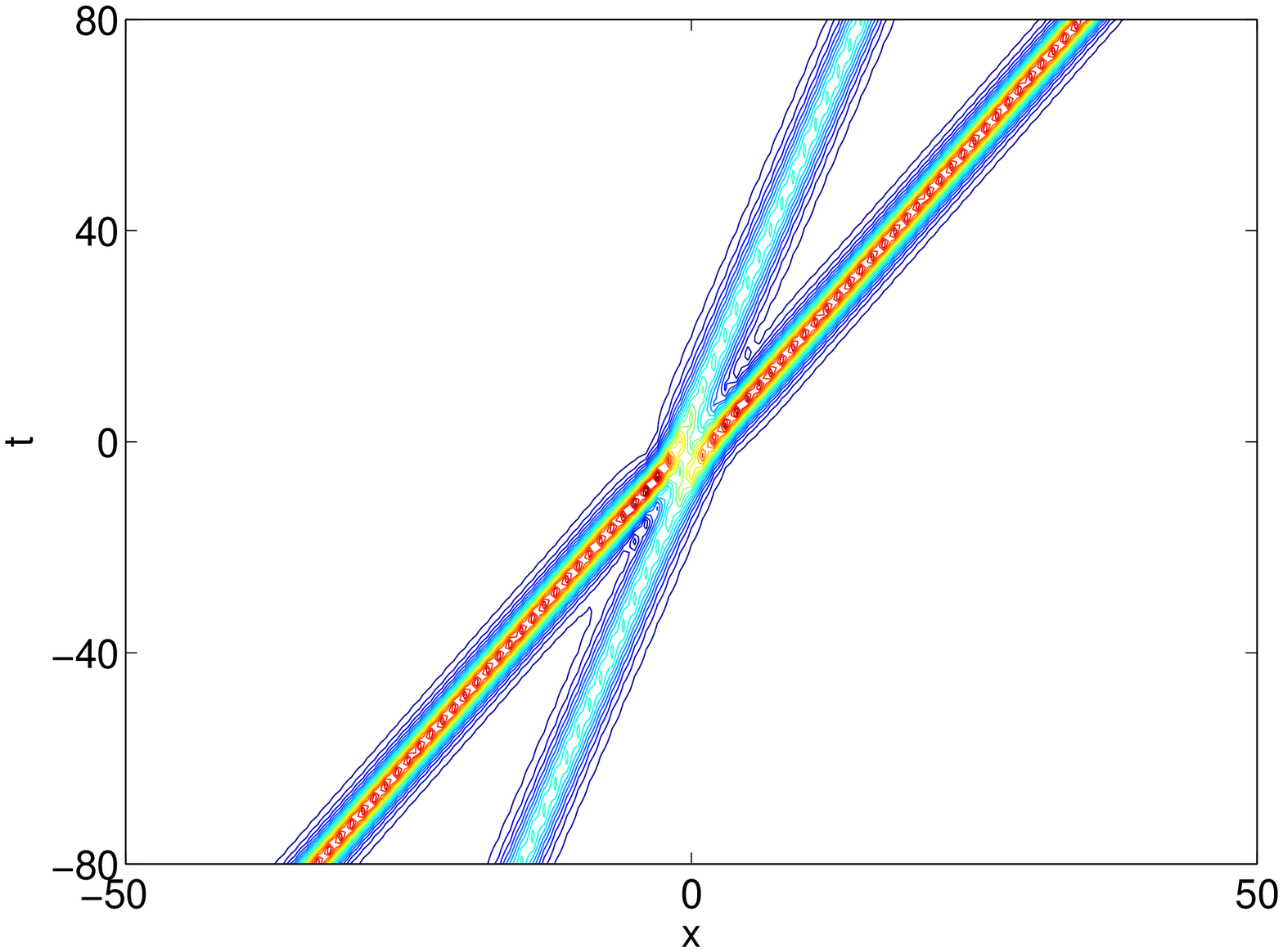}} \kern-0.315\textwidth
\hbox to
\textwidth{\hss(a)\kern6.5em\hss(b)\kern-1.5em} \kern+0.315\textwidth
\centerline{
\includegraphics[scale=0.35]{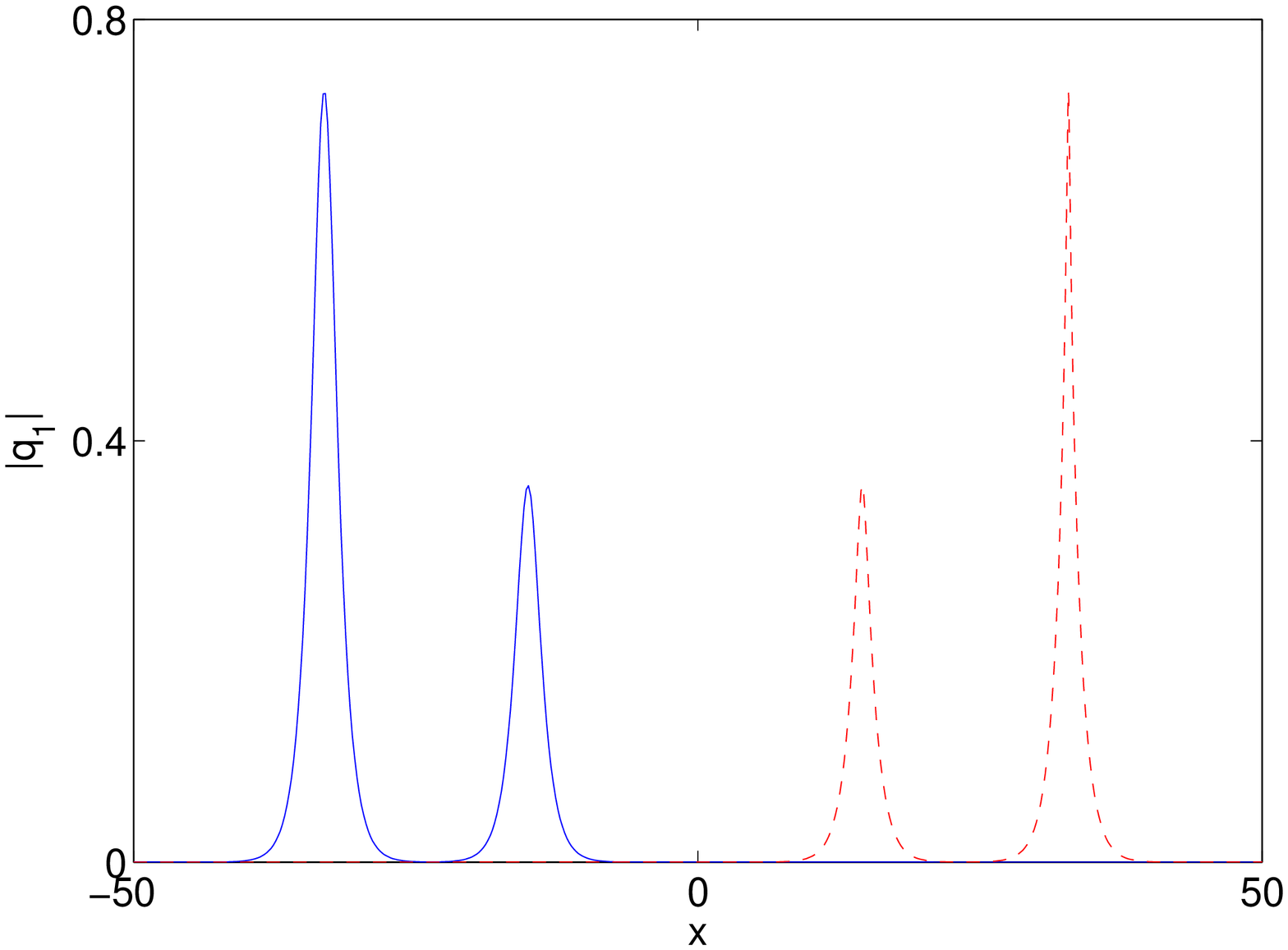}\quad
\includegraphics[scale=0.35]{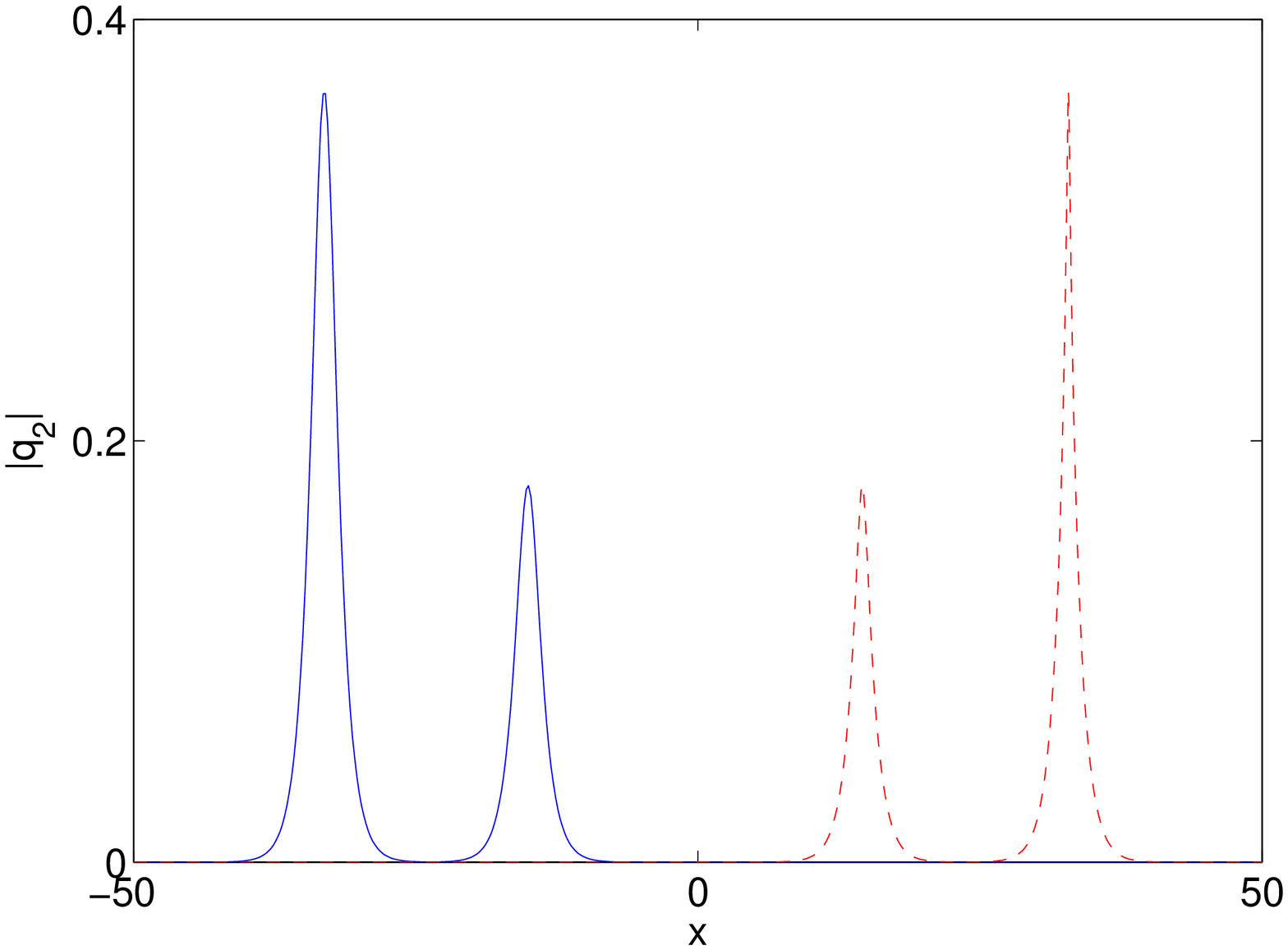}} \kern-0.315\textwidth
\hbox to
\textwidth{\hss(c)\kern6.5em\hss(d)\kern-1.5em} \kern+0.315\textwidth
\caption{Elastic collision in coupled complex short pulse equation.}
\label{f:elastic}
\end{figure}

It is interesting to note that if we just change the parameters in previous
two examples as $\alpha^{(1)}_{2}=0$, $\alpha^{(2)}_{2}=1.0$,
the energy of one soliton is concentrated in component $q_2$ before the collision. However, component $q_1$ gains some energy after the collision.
Such an example is shown in Fig. 6.

\begin{figure}[htbp]
\centerline{
\includegraphics[scale=0.35]{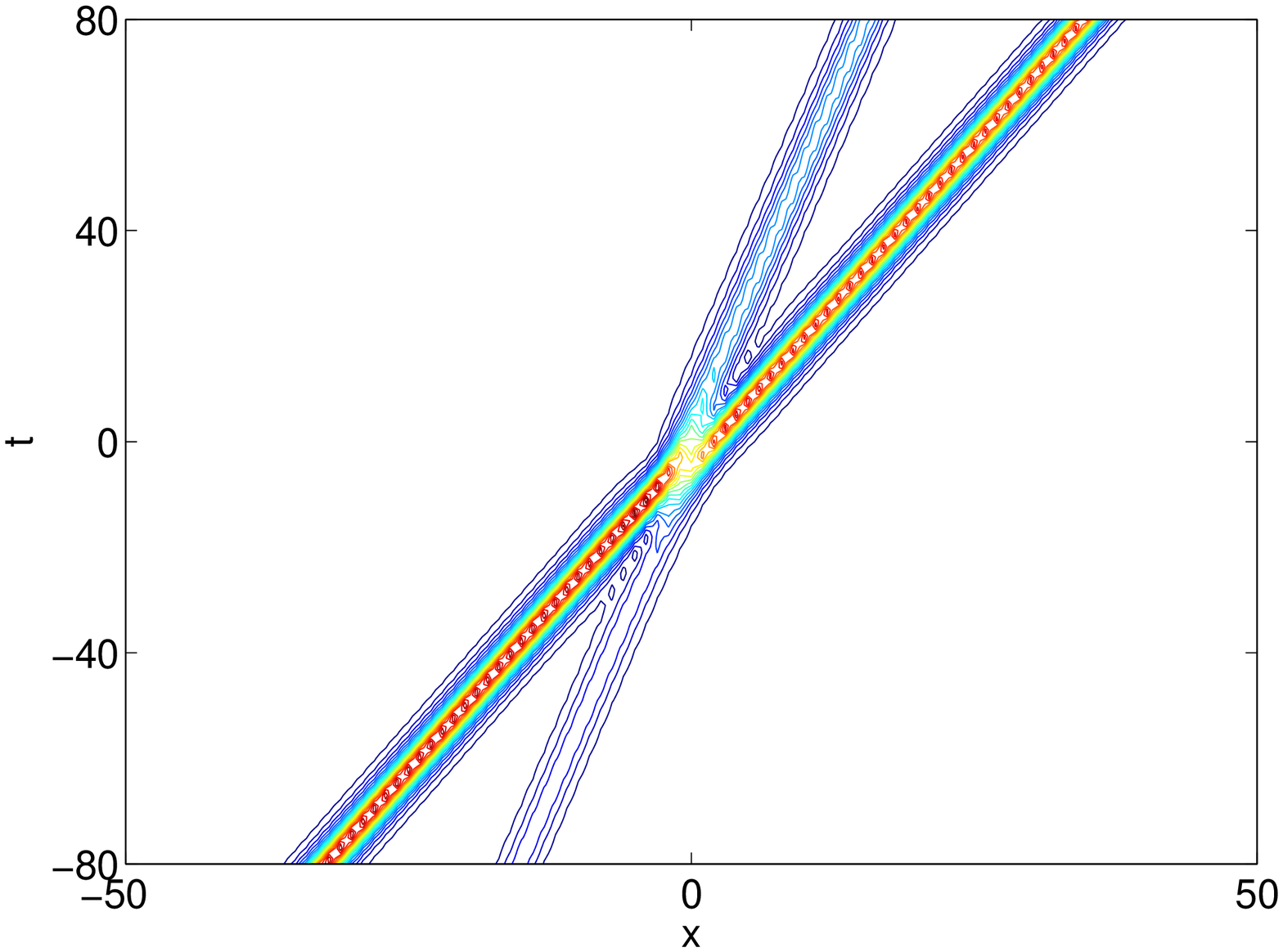}\quad
\includegraphics[scale=0.35]{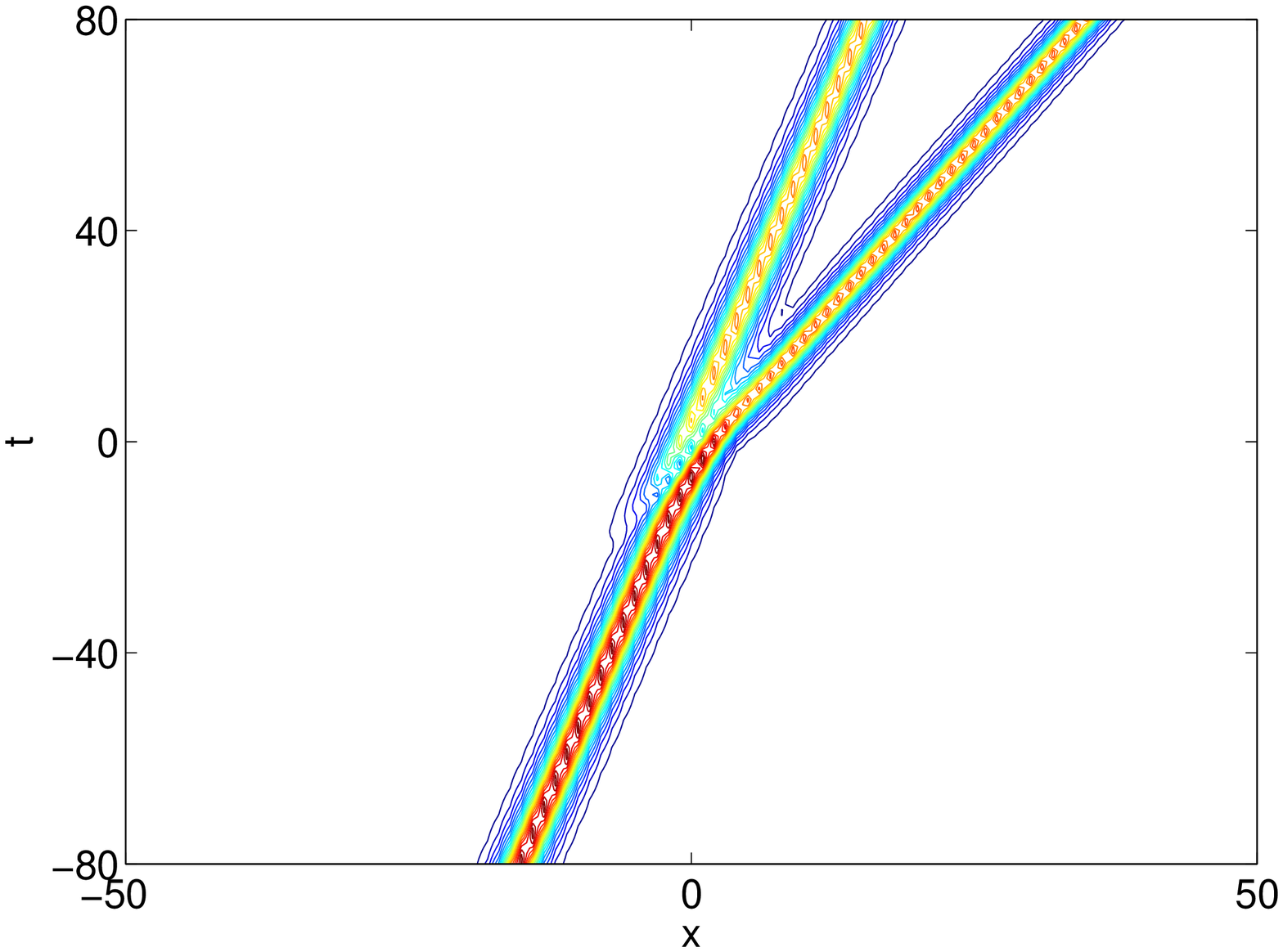}} \kern-0.315\textwidth
\hbox to
\textwidth{\hss(a)\kern6.5em\hss(b)\kern-1.5em} \kern+0.315\textwidth
\centerline{
\includegraphics[scale=0.35]{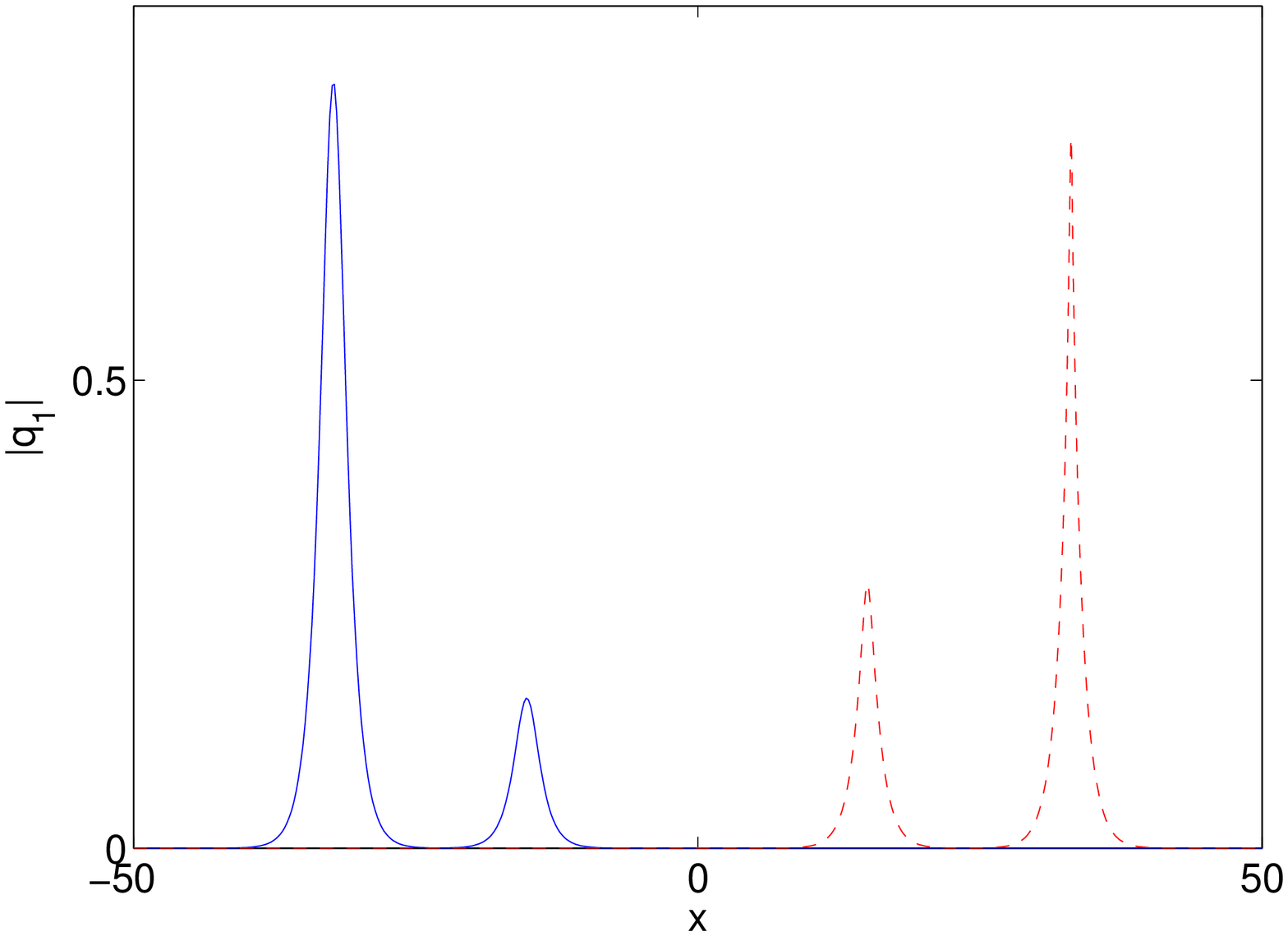}\quad
\includegraphics[scale=0.35]{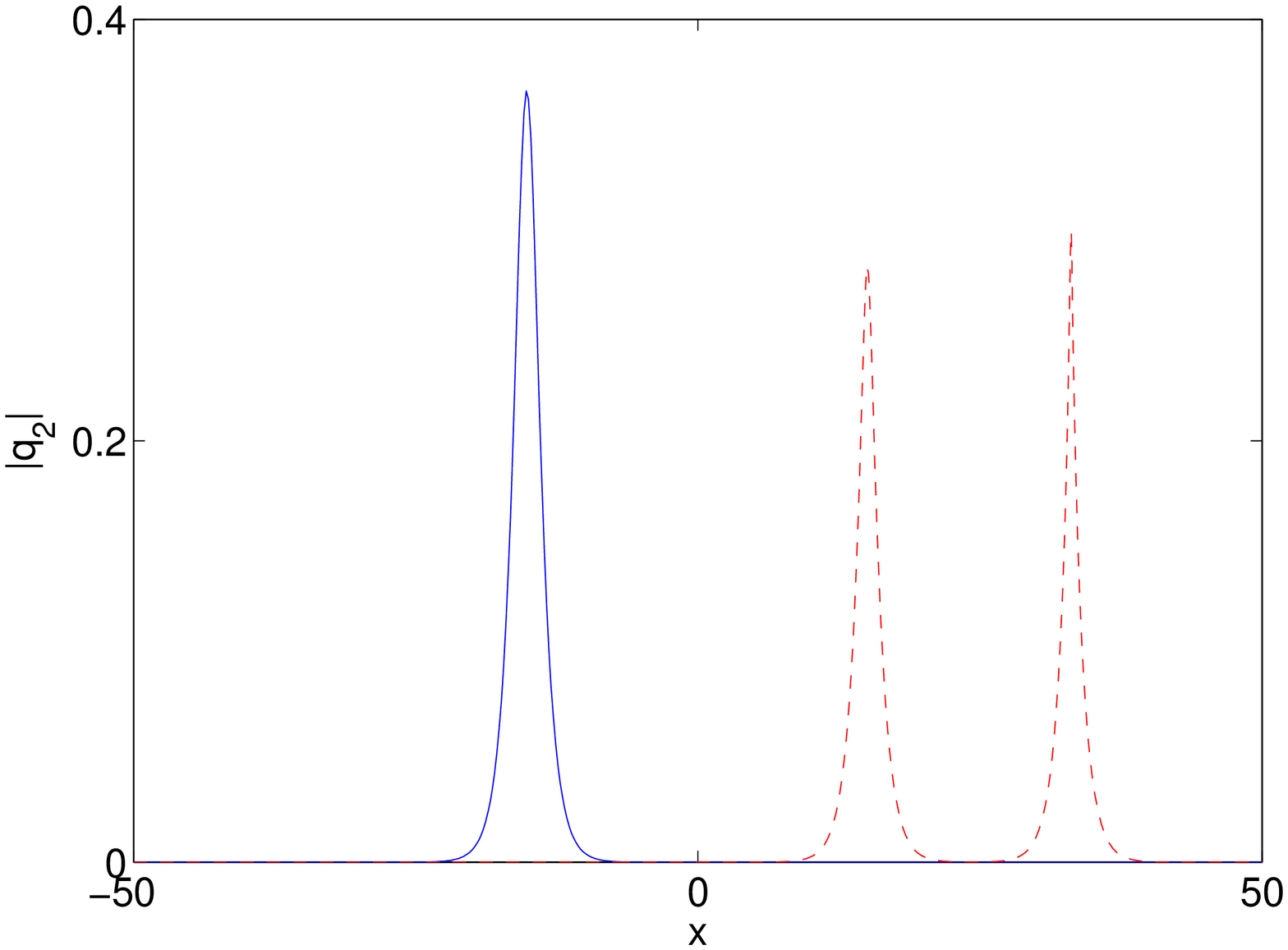}} \kern-0.315\textwidth
\hbox to
\textwidth{\hss(c)\kern6.5em\hss(d)\kern-1.5em} \kern+0.315\textwidth
\caption{Inelastic collision in coupled
complex short pulse equation for $p_{1}=1+1.2{\rm i}$, $p_{2}=1+2{\rm i}$, $\alpha^{(1)}_{1}=\alpha^{(2)}_{1}=1.0$, $\alpha^{(1)}_{2}=0$, $\alpha^{(2)}_{2}=1.0$. (a)-(b): contour plot; (c)-(d): profiles before and after the collision.}
\label{f:inelastic2}
\end{figure}
On the other hand, if we change the parameters as $\alpha^{(1)}_{2}=1.0$, $%
\alpha^{(2)}_{2}=0$, then the energy of one soliton, which are distributed between two components  before the
collision is concentrated into one component $q_2$ after the collision.
 The example is
shown in Fig. 7. 
\begin{figure}[htbp]
\centerline{
\includegraphics[scale=0.35]{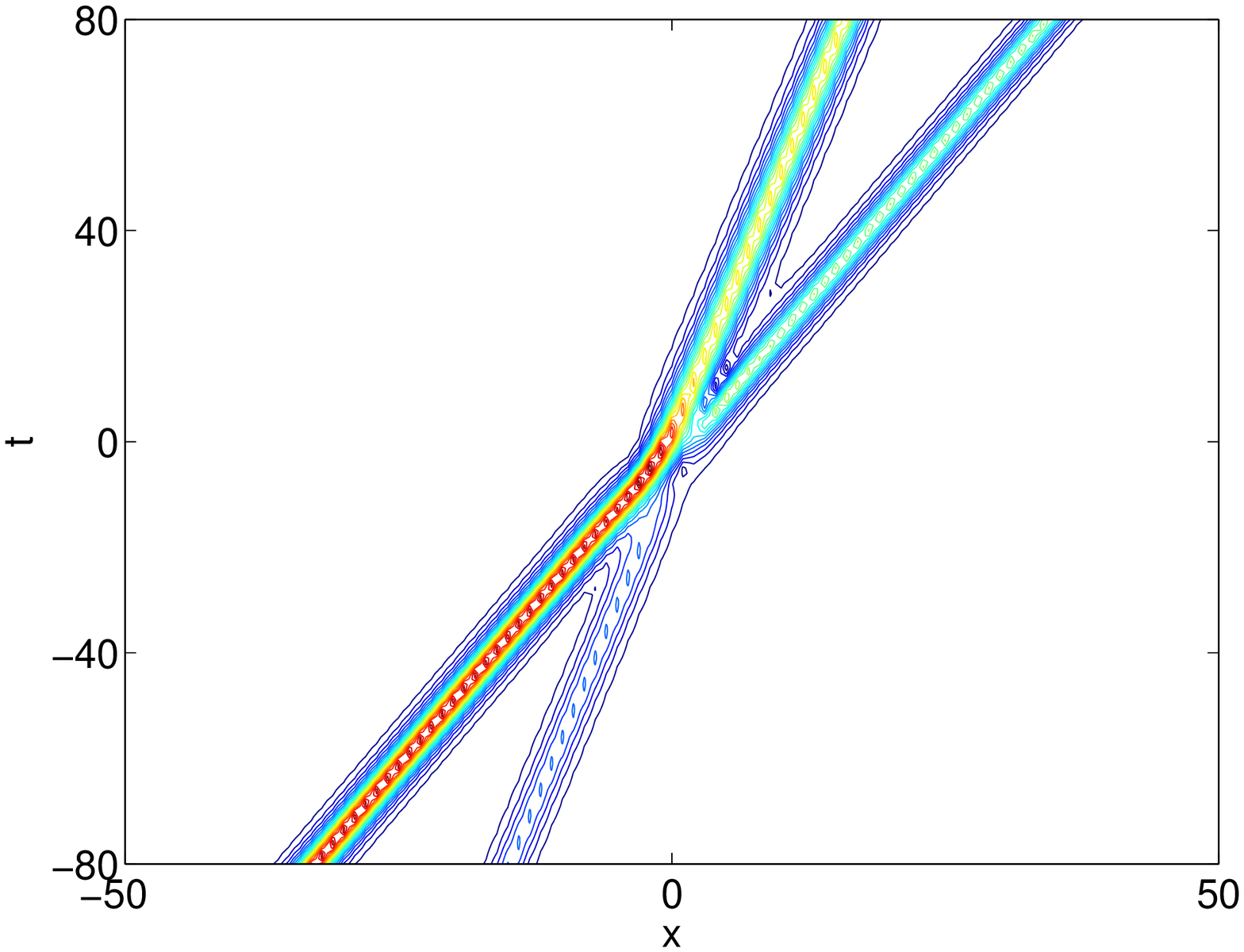}\quad
\includegraphics[scale=0.35]{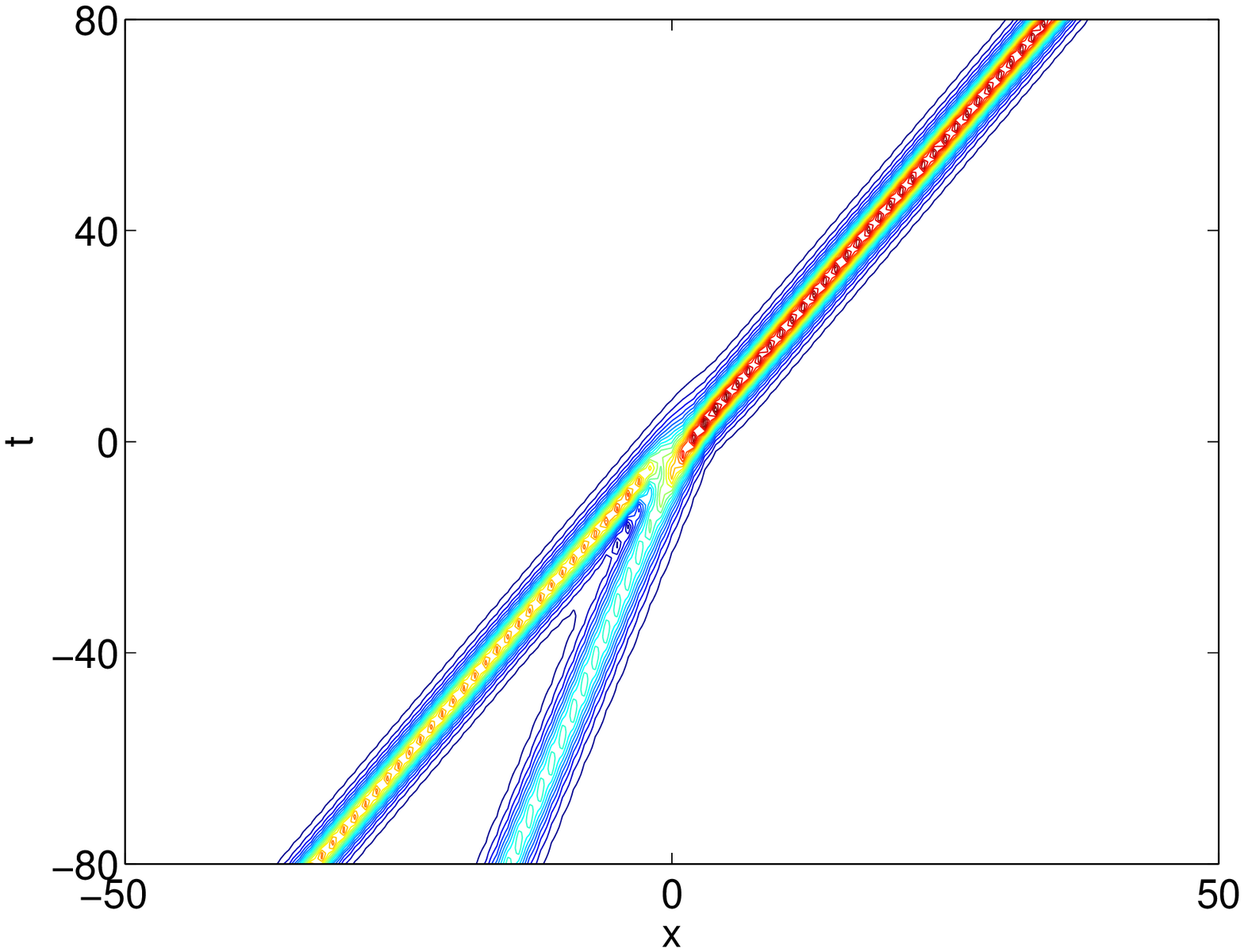}} \kern-0.315\textwidth
\hbox to
\textwidth{\hss(a)\kern6.5em\hss(b)\kern-1.5em} \kern+0.315\textwidth
\centerline{
\includegraphics[scale=0.35]{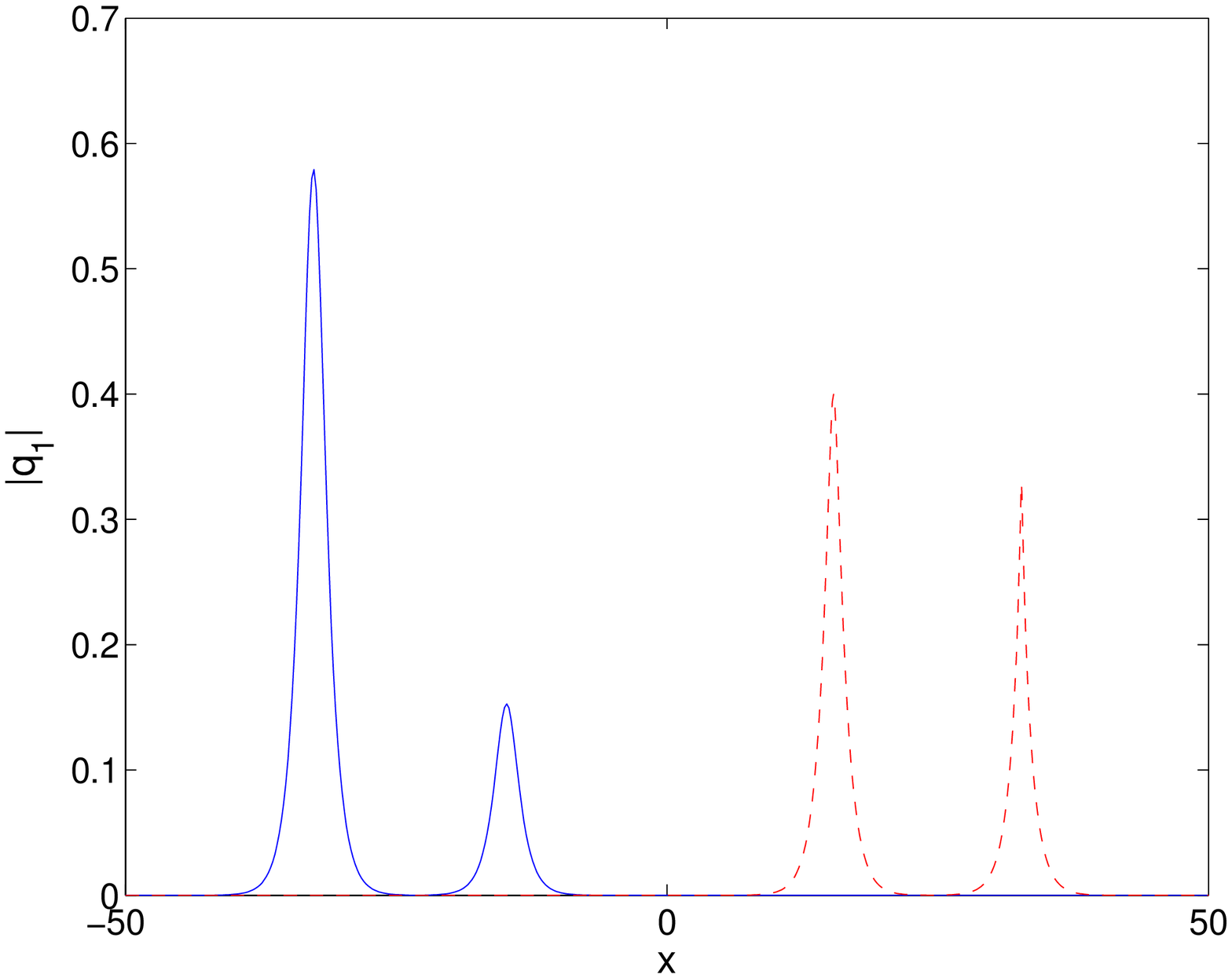}\quad
\includegraphics[scale=0.35]{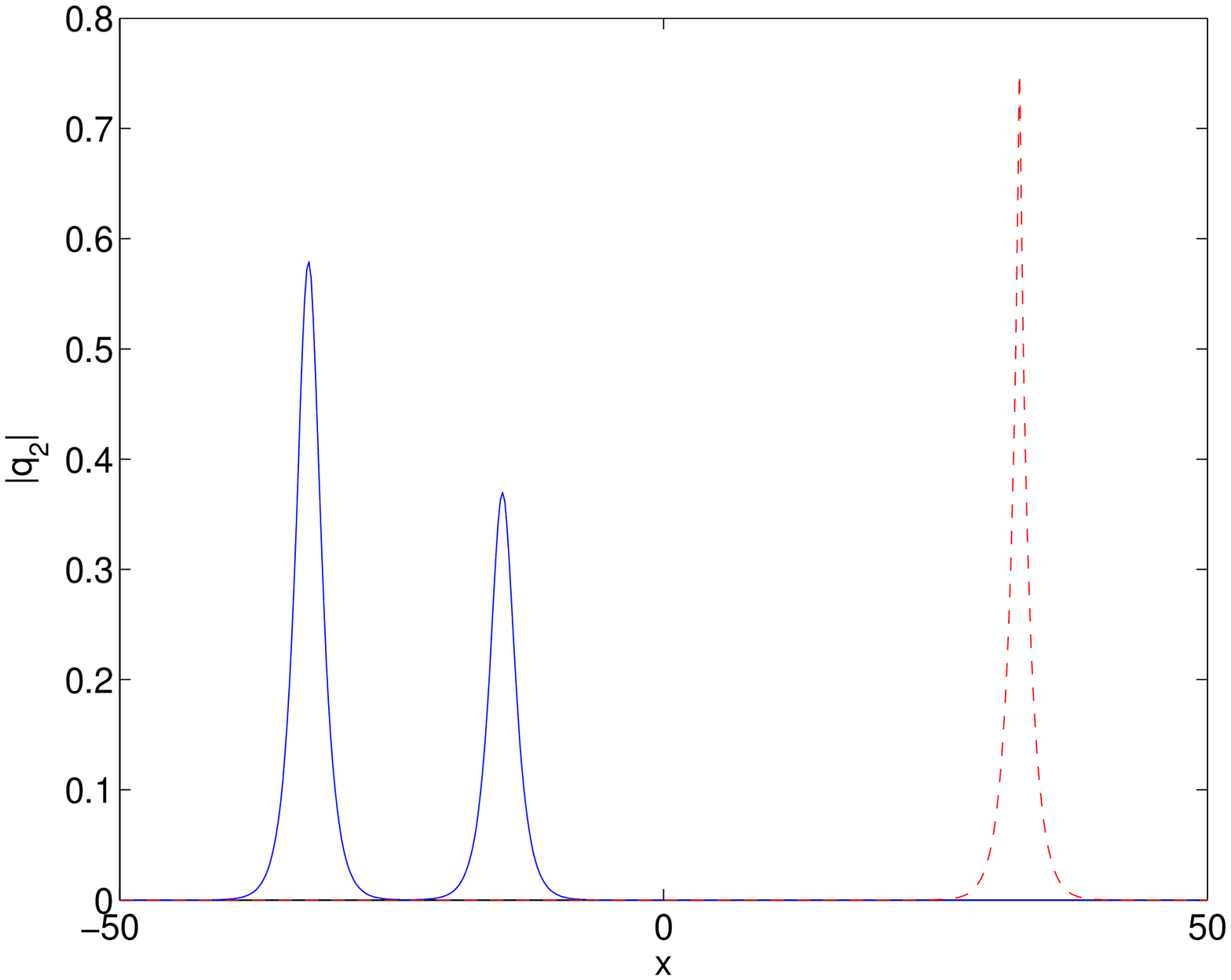}} \kern-0.315\textwidth
\hbox to
\textwidth{\hss(c)\kern6.5em\hss(d)\kern-1.5em} \kern+0.315\textwidth
\caption{Inelastic collision in coupled
complex short pulse equation for $p_{1}=1+1.2\mathrm{i}$, $p_{2}=1+2\mathrm{i%
}$, $\protect\alpha^{(1)}_{1}=\alpha^{(2)}_{1}=1.0$, \ $\alpha^{(1)}_{2}=1.0$, $\alpha^{(2)}_{2}=0$. (a)-(b): contour plot; (c)-(d): profiles before and after the collision.}
\label{f:inelastic3}
\end{figure}
\section{Concluding Remarks}
In this paper, we proposed a complex short pulse equation and its
two-component generalization. Both of the equations can be used to model
the propagation of ultra-short pulses in optical fibers. We have shown their
integrability by finding the Lax pairs and infinite numbers of conservation
laws. Furthermore, multi-soliton solutions are constructed via Hirota's
bilinear method. In particular, one-soliton solution for the CSP equation is
an envelope soliton with a few optical cycles under certain condition, which
perfectly match the requirement for the ultra-short pulses.
The $N$-solution for complex short pulse equation and its
two-component generalization is a benchmark for the study of soliton
interactions in ultra-short pulses propagation in optical fibers. It is expected that these analytical
solutions can be confirmed from experiments.

Similar to our previous results for the integrable discretizations of the short pulse equation \cite{SPE_discrete1}, how to construct integrable discretizations of the CSP and coupled CSP equations and how to
apply them for the numerical simulations is also an interesting topic to be
studied. It is obviously beyond the scope of the present paper, we
are to report the results on this aspect in a forthcoming paper.




\apptitle
\section{}
\appeqn
\textbf{Proof of Theorem 4.2}

\begin{proof}
First we define
\[
(b_j, \bar{\beta}_1)= \bar{\alpha}_j\delta_{\mu,1} \,, \quad (b_j, \bar{\beta%
}_2)= \bar{\alpha}_j\delta_{\mu,2}\,,
\]
where $index(b_j)=\mu$ , then from the fact
\[
\mathrm{Pf}(\bar{a}_j,a_k)= \mathrm{Pf} (a_{N+j},a_{N+k})\,, \mathrm{Pf}(%
\bar{b}_j,b_k)= \mathrm{Pf} (b_{N+j},b_{N+k})\,,
\]
we obtain
\[
\bar{f}=f\,, \quad \bar{g}= \mathrm{Pf} (d_0, \bar{\beta}_1, a_1, \cdots,
a_{2N}, b_1, \cdots, b_{2N})\,.
\]
Since
\[
\frac{\partial} {\partial y} \mathrm{Pf} (a_j,a_k)= (p_j -
p_k)e^{\eta_j+\eta_k} = \mathrm{Pf} (d_0, d_1, a_j,a_k)\,,
\]

\[
\frac{\partial} {\partial s} \mathrm{Pf} (a_j,a_k)= (p^{-1}_k - p^{-1}_j)
e^{\eta_j+\eta_k} = \mathrm{Pf} (d_{-1}, d_0, a_j,a_k)\,,
\]

\[
\frac{\partial^2} {\partial s^2} \mathrm{Pf} (a_j,a_k)= (p^{-2}_k -
p^{-2}_j) e^{\eta_j+\eta_k} = \mathrm{Pf} (d_{-2}, d_0, a_j,a_k)\,,
\]
\[
\frac{\partial^2} {\partial y \partial s}\mathrm{Pf} (a_j,a_k)= (p_jp^{-1}_k
- p_k p^{-1}_j) e^{\eta_j+\eta_k} = \mathrm{Pf} (d_{-1}, d_1, a_i,a_j)\,,
\]
we then have
\[
\frac{\partial f} {\partial y} = \mathrm{Pf} (d_0, d_1, \cdots)\,,
\]

\[
\frac{\partial f} {\partial s} = \mathrm{Pf} (d_{-1}, d_0, \cdots)\,,
\]

\[
\frac{\partial^2 f} {\partial s^2} = \mathrm{Pf} (d_{-2}, d_0, \cdots)\,,
\]

\[
\frac{\partial^2 f} {\partial y \partial s} = \mathrm{Pf} (d_{-1}, d_1,
\cdots)\,.
\]
Here $\mathrm{Pf} (d_0, d_1, a_1, \cdots, a_{2N}, b_1, \cdots, b_{2N})$ is
abbreviated by $\mathrm{Pf} (d_0, d_1, \cdots)$, so as other similar
pfaffians.

Furthermore, it can be shown
\begin{eqnarray*}
&& \frac{\partial g} {\partial y} = \frac{\partial} {\partial y} \left[%
\sum_{j=1}^{2N} (-1)^{j} \mathrm{Pf} (d_0, a_j) \mathrm{Pf} (\beta_1, \cdots
,\hat{a}_j, \cdots)\right] \\
&& =\sum_{j=1}^{2N} (-1)^{j} \left[ \left( {\partial_y} \mathrm{Pf} (d_0,
a_j) \right) \mathrm{Pf} (\beta_1, \cdots ,\hat{a}_j, \cdots) + \mathrm{Pf}
(d_0, a_j) {\partial_y} \mathrm{Pf} (\beta_1, \cdots ,\hat{a}_j, \cdots) %
\right] \\
&& =\sum_{j=1}^{2N} (-1)^{j} \left[ \mathrm{Pf} (d_1, a_j) \mathrm{Pf}
(\beta_1, \cdots ,\hat{a}_j, \cdots) + \mathrm{Pf} (d_0, a_j) \mathrm{Pf}
(\beta_1, d_0, d_1, \cdots ,\hat{a}_j, \cdots) \right] \\
&& = \mathrm{Pf} ( d_1, \beta_1, \cdots)+ \mathrm{Pf} ( d_0, \beta_1, d_0,
d_1, \cdots) \\
&& = \mathrm{Pf} (d_1, \beta_1, \cdots)\,.
\end{eqnarray*}

Here $\hat{a}_j$ means that the index $j$ is omitted. Similarly, we can show
\[
\frac{\partial g} {\partial s} = \mathrm{Pf} (d_{-1}, \beta_1, \cdots)\,,
\]

\begin{eqnarray*}
&& \frac{\partial^2 g} {\partial y \partial s} = \frac{\partial} {\partial y}
\left[\sum_{j=1}^{2N} (-1)^{j} \mathrm{Pf} (d_{-1}, a_j) \mathrm{Pf}
(\beta_1, \cdots ,\hat{a}_j, \cdots)\right] \\
&& =\sum_{j=1}^{2N} (-1)^{j} \left[ \left( {\partial_y} \mathrm{Pf} (d_{-1},
a_j) \right) \mathrm{Pf} (\beta_1, \cdots ,\hat{a}_j, \cdots) + \mathrm{Pf}
(d_{-1}, a_j) {\partial_y} \mathrm{Pf} (\beta_1, \cdots ,\hat{a}_j, \cdots) %
\right] \\
&& =\sum_{j=1}^{2N} (-1)^{j} \left[ \mathrm{Pf} (d_0, a_j) \mathrm{Pf}
(\beta_1, \cdots ,\hat{a}_j, \cdots) + \mathrm{Pf} (d_{-1}, a_j) \mathrm{Pf}
(\beta_1, d_0, d_1, \cdots ,\hat{a}_j, \cdots) \right] \\
&& = \mathrm{Pf} (d_0, \beta_1, \cdots)+ \mathrm{Pf} (d_{-1}, \beta_1, d_0,
d_1, \cdots)\,.
\end{eqnarray*}

An algebraic identity of pfaffian \cite{Hirota}
\begin{eqnarray*}
&& \mathrm{Pf} (d_{-1}, \beta_1, d_0, d_1, \cdots) \mathrm{Pf} (\cdots)=
\mathrm{Pf} (d_{-1}, d_0, \cdots) \mathrm{Pf} (d_1, \beta_1, \cdots) \\
&& \quad - \mathrm{Pf} (d_{-1}, d_1, \cdots) \mathrm{Pf} (d_0, \beta_1,
\cdots) + \mathrm{Pf} (d_{-1}, \beta_1, \cdots) \mathrm{Pf} (d_0, d_1,
\cdots)\,,
\end{eqnarray*}
implies
\[
( {\partial_s} {\partial_y} g-g) \times f = {\partial_s} f \times {\partial_y%
} g - {\partial_s} {\partial_y} f \times g + {\partial_s} g \times {%
\partial_y} f \,.
\]
Therefore, the first bilinear equation is approved.

The second bilinear equation can be proved in the same way by Iwao and
Hirota \cite{IwaoHirota}.
\begin{eqnarray}
&& \frac{\partial^2 f} {\partial s^2} \times 0 - \frac{\partial f} {\partial
s} \frac{\partial f} {\partial s}  \nonumber \\
&& = \mathrm{Pf} (d_{-2}, d_0, \cdots) \mathrm{Pf} (d_{0}, d_0, \cdots) -
\mathrm{Pf} (d_{-1}, d_0, \cdots) \mathrm{Pf} (d_{-1}, d_0, \cdots)
\nonumber \\
&& = \sum_{i=1}^{2N} (-1)^i \mathrm{Pf} (d_{-2}, a_i) \mathrm{Pf} (d_0,
\cdots, \hat{a}_i, \cdots) \sum_{j=1}^{2N} (-1)^j \mathrm{Pf} (d_{0}, a_j)
\mathrm{Pf} (d_0, \cdots, \hat{a}_j, \cdots)  \nonumber \\
&& -\sum_{i=1}^{2N} (-1)^i \mathrm{Pf} (d_{-1}, a_i) \mathrm{Pf} (d_0,
\cdots, \hat{a}_i, \cdots) \sum_{j=1}^{2N} (-1)^j \mathrm{Pf} (d_{-1}, a_j)
\mathrm{Pf} (d_0, \cdots, \hat{a}_j, \cdots)  \nonumber \\
&& =\sum_{i,j=1}^{2N} (-1)^{i+j} \left[ \mathrm{Pf} (d_{-2}, a_i) \mathrm{Pf}
(d_{0}, a_j) -\mathrm{Pf} (d_{-1}, a_i) \mathrm{Pf} (d_{-1}, a_j) \right]
\nonumber \\
&& \quad \times \mathrm{Pf} (d_0, \cdots, \hat{a}_i, \cdots) \mathrm{Pf}
(d_0, \cdots, \hat{a}_j, \cdots)  \nonumber \\
&&=\sum_{i,j=1}^{2N} (-1)^{i+j+1} \left[p_i^{-2} + p_i^{-1}p_j^{-1} \right]
\mathrm{Pf} (a_i, a_j) \mathrm{Pf} (d_0, \cdots, \hat{a}_i, \cdots) \mathrm{%
Pf} (d_0, \cdots, \hat{a}_j, \cdots)  \nonumber
\end{eqnarray}

The summation over the second term within the bracket vanishes due to the
fact that
\begin{eqnarray*}
&& \sum_{i,j=1}^{2N} (-1)^{i+j+1} p_i^{-1}p_j^{-1} \mathrm{Pf} (a_i, a_j)
\mathrm{Pf} (d_0, \cdots, \hat{a}_i, \cdots) \mathrm{Pf} (d_0, \cdots, \hat{a%
}_j, \cdots) \\
&& = \sum_{j,i=1}^{2N} (-1)^{j+i+1} p_j^{-1}p_i^{-1} \mathrm{Pf} (a_j, a_i)
\mathrm{Pf} (d_0, \cdots, \hat{a}_j, \cdots) \mathrm{Pf} (d_0, \cdots, \hat{a%
}_i, \cdots) \\
&& = -\sum_{i,j=1}^{2N} (-1)^{i+j+1} p_i^{-1}p_j^{-1} \mathrm{Pf} (a_i, a_j)
\mathrm{Pf} (d_0, \cdots, \hat{a}_i, \cdots) \mathrm{Pf} (d_0, \cdots, \hat{a%
}_j, \cdots)\,.
\end{eqnarray*}
Therefore,
\begin{eqnarray}
&& - \frac{\partial f} {\partial s} \frac{\partial f} {\partial s}
=\sum_{i,j=1}^{2N} (-1)^{i+j+1} p_i^{-2} \mathrm{Pf} (a_i, a_j) \mathrm{Pf}
(d_0, \cdots, \hat{a}_i, \cdots) \mathrm{Pf} (d_0, \cdots, \hat{a}_j, \cdots)
\nonumber \\
&&=\sum_{i=1}^{2N} (-1)^{i+1} p_i^{-2} \mathrm{Pf} (d_0, \cdots, \hat{a}_i,
\cdots) \left[ \sum_{j=1}^{2N} (-1)^{j} \mathrm{Pf} (a_i, a_j) \mathrm{Pf}
(d_0, \cdots, \hat{a}_j, \cdots) \right]  \nonumber  \label{CSP1_proof5}
\end{eqnarray}
Further, we note that the following identity can be substituted into the
term within bracket
\begin{eqnarray*}
&& \sum_{j=1}^{2N} (-1)^{j} \mathrm{Pf} (a_i, a_j) \mathrm{Pf} (d_0, \cdots,
\hat{a}_j, \cdots) \\
&& = \mathrm{Pf} (d_{0}, a_i) \mathrm{Pf} (\cdots) + (-1)^{i+1} \mathrm{Pf}
(d_0, \cdots, \hat{b}_i, \cdots)\,
\end{eqnarray*}
which is obtained from the expansion of the following vanishing pfaffian $%
\mathrm{Pf} (a_i, d_0, \cdots)$ on $a_i$. Consequently, we have
\begin{eqnarray}
&& - \frac{\partial f} {\partial s} \frac{\partial f} {\partial s} =
\nonumber \\
&& \sum_{i=1}^{2N} (-1)^{i+1} p_i^{-2} \mathrm{Pf} (d_0, \cdots, \hat{a}_i,
\cdots) \left[\mathrm{Pf} (d_{0}, a_i) \mathrm{Pf} (\cdots) + (-1)^{i+1}
\mathrm{Pf} (d_0, \cdots, \hat{b}_i, \cdots)\right]\,,  \nonumber \\
&& = -\mathrm{Pf} (\cdots) \mathrm{Pf} (d_{-2}, d_0, \cdots)+
\sum_{i=1}^{2N} p_i^{-2} \mathrm{Pf} (d_0, \cdots, \hat{a}_i, \cdots)
\mathrm{Pf} (d_0, \cdots, \hat{b}_i, \cdots)\,,
\end{eqnarray}
which can be rewritten as
\begin{equation}
\frac{\partial^2 f} {\partial s^2} f- \frac{\partial f} {\partial s} \frac{%
\partial f} {\partial s}= \sum_{i=1}^{2N} p_i^{-2} \mathrm{Pf} (d_0, \cdots,
\hat{a}_i, \cdots) \mathrm{Pf} (d_0, \cdots, \hat{b}_i, \cdots)\,.
\end{equation}

Now, we work on the r.h.s of the second bilinear equation.
\begin{eqnarray}  \label{CSP1_proof1}
&& \frac 12 |g|^2 = \frac 12 \mathrm{Pf} (d_0,\beta_1, \cdots) \mathrm{Pf}
(d_0, \bar{\beta}_1, \cdots)  \nonumber \\
&& = \frac 12 \sum_{i,j}^{2N} (-1)^{i+j}\mathrm{Pf} (b_i, \beta_1) \mathrm{Pf%
} (d_0, \cdots,\hat{b}_i, \cdots) \mathrm{Pf} (b_j,\bar{\beta}_1) \mathrm{Pf}
(d_0, \cdots,\hat{b}_j, \cdots)  \nonumber \\
&& = \frac 14 \sum_{i,j}^{2N} (-1)^{i+j} (\alpha_i \bar{\alpha}_j) \mathrm{Pf%
} (d_0, \cdots,\hat{b}_i, \cdots) \mathrm{Pf} (d_0, \cdots,\hat{b}_j, \cdots)
\nonumber \\
&& = \sum_{i,j}^{2N} (-1)^{i+j} \left(p_i^{-2}-p_{j}^{-2}\right) \mathrm{Pf}
(b_i,b_j)\mathrm{Pf}(d_0, \cdots,\hat{b}_i, \cdots) \mathrm{Pf} (d_0, \cdots,%
\hat{b}_j, \cdots)  \nonumber \\
\end{eqnarray}
Next, the expansion of the vanishing pfaffian $\mathrm{Pf} (b_i, d_0,
\cdots) $ on $b_i$ yields
\begin{equation}
\sum_{j=1}^{2N} (-1)^{i+j}\mathrm{Pf} (b_i, b_j) \mathrm{Pf} (d_0, \cdots,
\hat{b}_j, \cdots) = \mathrm{Pf} (d_0, \cdots,\hat{a}_i, \cdots)\,,
\end{equation}
which subsequently leads to
\begin{eqnarray}
&& \sum_{i,j}^{2N} (-1)^{i+j} p_i^{-2} \mathrm{Pf} (b_i,b_j)\mathrm{Pf}(d_0,
\cdots,\hat{b}_i, \cdots) \mathrm{Pf} (d_0, \cdots,\hat{b}_j, \cdots)
\nonumber \\
&& = \sum_{i}^{2N} p_i^{-2} \mathrm{Pf} (d_0, \cdots,\hat{a}_i, \cdots)
\mathrm{Pf} (d_0, \cdots,\hat{b}_i, \cdots)\,.  \label{CSP1_proof2}
\end{eqnarray}
Similarly, we can show that
\begin{eqnarray}
&& -\sum_{i,j}^{2N} (-1)^{i+j} p_j^{-2} \mathrm{Pf} (b_i,b_j)\mathrm{Pf}%
(d_0, \cdots,\hat{b}_i, \cdots) \mathrm{Pf} (d_0, \cdots,\hat{b}_j, \cdots)
\nonumber \\
&& = \sum_{j}^{2N} p_j^{-2} \mathrm{Pf} (d_0, \cdots,\hat{a}_j, \cdots)
\mathrm{Pf} (d_0, \cdots,\hat{b}_j, \cdots)\,.  \label{CSP1_proof3}
\end{eqnarray}
Substituting Eqs. (\ref{CSP1_proof2})--(\ref{CSP1_proof2}) into Eq. (\ref%
{CSP1_proof1}), we arrive at
\begin{equation}  \label{CSP1_proof4}
\frac 12 |g|^2= 2\sum_{i}^{2N} p_i^{-2} \mathrm{Pf} (d_0, \cdots,\hat{a}_i,
\cdots) \mathrm{Pf} (d_0, \cdots,\hat{b}_i, \cdots)\,.
\end{equation}
Consequently we have
\begin{equation}
2\frac{\partial^2 f} {\partial s^2} f- 2\frac{\partial f} {\partial s} \frac{%
\partial f} {\partial s}= \frac 12 |g|^2\,,
\end{equation}
which is nothing but the second bilinear equation. Therefore, the proof is
complete.
\end{proof}

\textbf{The proof of Theorem 4.6}

\begin{proof}
The proof of the first bilinear equation can be done exactly in the same way
as for the complex short pulse equation. In what follows, we prove the
second equation by starting from the r.h.s of this equation. Because
\[
\bar{g}_1= \mathrm{Pf} (d_0, \bar{\beta}_1, a_1, \cdots, a_{2N}, b_1,
\cdots, b_{2N})\,,
\]
\[
\bar{g}_2= \mathrm{Pf} (d_0, \bar{\beta}_2, a_1, \cdots, a_{2N}, b_1,
\cdots, b_{2N})\,,
\]
the r.h.s of the bilinear equation turns out to be
\begin{eqnarray}  \label{CCSP1_proof1}
&& \frac 12 \left( g_{1} \bar{g}_{1} + g_{2} \bar{g}_{2} \right)  \nonumber
\\
&& = \frac 12 \sum^2_{k=1} \sum_{i,j}^{2N} (-1)^{i+j}\mathrm{Pf} (b_i,
\beta_k) \mathrm{Pf} (d_0, \cdots,\hat{b}_i, \cdots) \mathrm{Pf} (b_j,\bar{%
\beta}_k) \mathrm{Pf} (d_0, \cdots,\hat{b}_j, \cdots)  \nonumber \\
&& = \frac 14 \sum_{i,j}^{2N} (-1)^{i+j} \sum^2_{k=1}(\alpha^{(k)}_i \bar{%
\alpha}^{(k)}_j) \mathrm{Pf} (d_0, \cdots,\hat{b}_i, \cdots) \mathrm{Pf}
(d_0, \cdots,\hat{b}_j, \cdots)  \nonumber \\
&& = \sum_{i,j}^{2N} (-1)^{i+j} \left(p_i^{-2}-p_{j}^{-2}\right) \mathrm{Pf}
(b_i,b_j)\mathrm{Pf}(d_0, \cdots,\hat{b}_i, \cdots) \mathrm{Pf} (d_0, \cdots,%
\hat{b}_j, \cdots)  \nonumber \\
\end{eqnarray}
Similar to the complex short pulse equation, we can show 

\begin{equation}  \label{CCSP1_proof4}
\frac 12 \left(|g_{1}|^2 + |g_{2}|^2 \right) = 2\sum_{i}^{2N} p_i^{-2}
\mathrm{Pf} (d_0, \cdots,\hat{a}_i, \cdots) \mathrm{Pf} (d_0, \cdots,\hat{b}%
_i, \cdots)\,.
\end{equation}
Regarding the r.h.s of the bilinear equation, exactly the same as the proof
of the Theorem 4.2, we have
\begin{equation}  \label{CCSP1_proof5}
\frac{\partial^2 f} {\partial s^2} f- \frac{\partial f} {\partial s} \frac{%
\partial f} {\partial s} =\sum_{i}^{2N} p_i^{-2} \mathrm{Pf} (d_0, \cdots,%
\hat{a}_i, \cdots) \mathrm{Pf} (d_0, \cdots,\hat{b}_i, \cdots)\,.
\end{equation}
Therefore the second bilinear equation is proved.
\end{proof}

\thank
\section{}
The author is grateful for the useful discussions with Dr. Yasuhiro Ohta
(Kobe University) and Dr. Kenichi Maruno at Waseda University. This work is partially supported by the National Natural Science Foundation of China (No. 11428102).



\end{document}